\documentclass[preprint,12pt]{elsarticle}
\usepackage{lineno}

\usepackage[T1]{fontenc}
\usepackage[latin2]{inputenc}
\usepackage[english]{babel}
\usepackage{dsfont, complexity} 
\usepackage{amsmath,amsthm}
\usepackage{amssymb}
\usepackage{float, lmodern}
\usepackage{graphicx}
\usepackage{tikz}
\usepackage{xcolor}
\usepackage{paralist}
\usepackage{refcount}
\newcommand{\myqed}{$\blacksquare$}
\newcommand{\myproof}{\noindent\textit{Proof. }}

\usepackage{mathtools}

\usetikzlibrary{arrows,decorations.markings,patterns,calc, positioning, backgrounds}

\tikzstyle{vertex} = [circle, draw=black, scale=0.7]
\tikzstyle{edgelabel} = [circle, fill=white, scale=0.8]

\newcommand{\ktwotwo}[2]{
\coordinate (p) at #1;
\begin{scope}[rotate around={#2:(p)}]
	\node[vertex, ultra thick, gray] (p) at (p) {$$};
	\node[vertex] (p') at ($(1,0)+(p)$) {$$};
	\node[vertex] (q) at ($(0, -1)+(p)$) {$$};
	\node[vertex] (q') at ($(1, -1)+(p)$) {$$};
	
	\draw [ultra thick, gray, dotted] (p) -- (p');
	\draw [] (p) -- (q');
	\draw [] (q) -- (q');
	\draw [] (q) -- (p');
\end{scope}
}
\newtheorem{theorem}{Theorem}[section]
\newtheorem{claim}[theorem]{Claim}
\newtheorem{pr}[theorem]{Problem}

\newtheorem{conj}[theorem]{Conjecture}

\newtheorem{remark}[theorem]{Remark}
\newcommand{\true}{\mathsf{true}}
\newcommand{\false}{\mathsf{false}}

\journal{Discrete Optimization}
\begin{document}
\begin{frontmatter}

\title{Stable Marriage and Roommates problems with restricted edges: Complexity and approximability\tnoteref{title}}

\tnotetext[title]{A preliminary version of this paper appeared in the Proceedings of SAGT 2015: the 8th International Symposium on Algorithmic Game Theory.}

\author[cseh]{\'Agnes Cseh}
\ead{cseh@math.tu-berlin.de}
\fntext[cseh]{Supported by COST Action IC1205 on Computational Social Choice and by the Deutsche Telekom Stiftung. Part of this work was carried out whilst visiting the University of Glasgow. Present address:  School of Computer Science, Reykjavik University, Menntavegur 1, 101 Reykjavik, Iceland.}

\author[manlove]{David F. Manlove}
\fntext[manlove]{Supported by Engineering and Physical Sciences Research Council grant EP/K010042/1.}
\ead{David.Manlove@glasgow.ac.uk}

\address[cseh]{Institute for Mathematics, Technische Universit\"at Berlin, Sekr. MA 5-2, Stra{\ss}e des 17. Juni 136, 10623 Berlin, Germany}

\address[manlove]{School of Computing Science, Sir Alwyn Williams Building, University of Glasgow, Glasgow G12 8QQ, UK}

\begin{abstract}
In the Stable Marriage and Roommates problems, a set of agents is given, each of them having a strictly ordered preference list over some or all of the other agents. A matching is a set of disjoint~pairs of mutually acceptable agents. If any two agents mutually prefer each other to their partner, then they block the matching, otherwise, the matching is said to be stable. We investigate the complexity of finding a solution satisfying additional constraints on restricted pairs of agents. Restricted pairs can be either \emph{forced} or \emph{forbidden}. A stable solution must contain all of the forced pairs, while it must contain none of the forbidden pairs.

Dias et al.~\cite{DFFS03} gave a polynomial-time algorithm  to decide whether such a solution exists in the presence of restricted edges. If the answer is no, one might look for a solution close to optimal. Since optimality in this context means that the matching is stable and satisfies all constraints on restricted pairs, there are two ways of relaxing the constraints by permitting a solution to: (1)~be blocked by as few as possible pairs, or (2)~violate as few as possible constraints on restricted pairs.

Our main theorems prove that for the (bipartite) Stable Marriage problem, case~(1) leads to $\NP$-hardness and inapproximability results, whilst case~(2) can be solved in polynomial time.  For non-bipartite Stable Roommates instances, case~(2) yields an $\NP$-hard 
problem.  In the case of $\NP$-hard problems, we also discuss polynomially solvable special cases, arising from restrictions on the lengths of the preference lists, or upper bounds on the numbers of restricted pairs.
\end{abstract}

\end{frontmatter}

\section{Introduction}
\label{se:intro}

In the classical \emph{Stable Marriage problem} ({\sc sm})~\cite{GS62}, a bipartite graph is given, where one colour class symbolises a set of men $U$ and the other colour class stands for a set of women~$W$. Man $u$ and woman $w$ are connected by edge $uw$ if they find one another mutually acceptable. Each participant provides a strictly ordered preference list of the acceptable agents of the opposite gender. An edge $uw$ \emph{blocks} matching $M$ if it is not in $M$, but each of $u$ and $w$ is either unmatched or prefers the other to their partner. A \emph{stable matching} is a matching not blocked by any edge. From the seminal paper of Gale and Shapley~\cite{GS62}, we know that the existence of such a stable solution is guaranteed and one can be found in linear time. Moreover, the solutions form a distributive lattice~\cite{Knu76}. The two extreme points of this lattice are called the \emph{man-} and \emph{woman-optimal stable matchings}~\cite{GS62}. These assign each man/woman their best partner reachable in any stable matching. Another interesting and useful property of stable solutions is the so-called Rural Hospitals Theorem. Part of this theorem states that if an agent is unmatched in one stable matching, then all stable solutions leave him unmatched~\cite{GS85}.

One of the most widely studied  extensions of {\sc sm} is the \emph{Stable Roommates problem} ({\sc sr})~\cite{GS62,Irv85}, defined on general graphs instead of bipartite graphs. The notion of a blocking edge is as defined above (except that it can now involve any two agents in general), but several results do not carry over to this setting. For instance, the existence of a stable solution is not guaranteed any more. On the other hand, there is a linear-time algorithm to find a stable matching or report that none exists~\cite{Irv85}. Moreover, the corresponding variant of the Rural Hospitals Theorem holds in the roommates case as well: the set of matched agents is the same for all stable solutions~\cite{GI89}. We summarise this observation as follows:
\begin{theorem}[Gusfield and Irving~\cite{GI89}]
Given an instance of {\sc sr}, the same set of agents is matched in all stable matchings.
\label{th:rural}
\end{theorem}

Both {\sc sm} and {\sc sr} are widely used in various applications. In markets where the goal is to maximise social welfare instead of profit, the notion of stability is especially suitable as an optimality criterion~\cite{Rot84}. For {\sc sm}, the oldest and most common area of applications is employer allocation markets~\cite{RS90}.  On one side, job applicants are represented, while the job openings form the other side. Each application corresponds to an edge in the bipartite graph. The employers rank all applicants to a specific job offer and similarly, each applicant sets up a preference list of jobs. Given a proposed matching $M$ of applicants to jobs, if an employer--applicant pair exists such that the position is not filled or a worse applicant is assigned to it, and the applicant received no contract or a worse contract, then this pair blocks~$M$. In this case the employer and applicant find it mutually beneficial to enter into a contract outside of $M$, undermining its integrity. If no such blocking pair exists, then $M$ is stable. Stability as an underlying concept is also used to allocate graduating medical students to hospitals in many countries~\cite{Rot08}. {\sc sr} on the other hand has applications in the area of P2P networks~\cite{GLMMRV07}.

Forced and forbidden edges in {\sc sm} and {\sc sr} open the way to formulate various special requirements on the sought solution. Such edges now form part of the extended problem instance: if an edge is \emph{forced}, it must belong to a constructed stable matching, whilst if an edge is \emph{forbidden}, it must not.
In certain market situations, a contract is for some reason particularly important, or to the contrary, not wished by the majority of the community or by the central authority in control. In such cases, forcing or forbidding the edge and then seeking a stable solution ensures that the wishes on these specific contracts are fulfilled while stability is guaranteed. Henceforth, the term \emph{restricted edge} will be used to refer either to a forbidden edge or a forced edge. The remaining edges of the graph are referred as \emph{unrestricted edges}.

Note that simply deleting forbidden edges or fixing forced edges and searching for a stable matching on the remaining instance does not solve the problem of finding a stable matching with restricted edges. Deleted edges (corresponding to forbidden edges, or those adjacent to forced edges) can block that matching. Therefore, to meet both requirements on restricted edges and stability, more sophisticated methods are needed.

The attention of the community was drawn very early to the characterization of stable matchings that must contain a prescribed set of edges. In the seminal book of Knuth~\cite{Knu76}, forced edges first appeared under the term \emph{arranged marriages}. Knuth presented an algorithm that finds a stable matching with a given set of forced edges or reports that none exists, given an instance of {\sc sm}. This method runs in $O(n^2)$ time, where $n$ denotes the number of vertices in the graph. Gusfield and Irving~\cite{GI89} provided an algorithm for {\sc sm} based on rotations that terminates in $O(|Q|^2)$ time, following $O(n^4)$ pre-processing time, where $Q$ is the set of forced edges. This latter method is favoured over Knuth's if multiple forced sets of small cardinality are proposed.

Forbidden edges appeared only in~2003 in the literature, and were first studied by Dias et al.~\cite{DFFS03}. In their paper, complete bipartite graphs were considered, but the methods can easily be extended to incomplete preference lists. Their main result was the following (in the following theorem, and henceforth, $m$ is the total number of edges in the graph).

\begin{theorem}[Dias et al.~\cite{DFFS03}] 
\label{th:bip_decision} 
The problem of finding a stable matching in a  {\sc sm} instance with forced and forbidden edges or reporting that none exists is solvable in $O(m)$ time.
\end{theorem}

While Knuth's method relies on basic combinatorial properties of stable matchings, the other two algorithms make use of \emph{rotations}. We refer the reader to~\cite{GI89} for background on these. The problem of finding a stable matching with forced and forbidden edges in an {\sc sm} instance can easily be formulated as a weighted stable matching problem (that is, we seek a stable matching with minimum weight, where the weight of a matching $M$ is the sum of the weights of the edges in~$M$). Let us assign all forced edges weight~$-1$, all forbidden edges weight~$1$, and all remaining edges weight~0. A stable matching satisfying all constraints on restricted edges exists if and only if there is a stable matching of weight~$-|Q|$ in the weighted instance, where $Q$ is the set of forced edges. With the help of rotations, minimum weight stable matchings can be found in polynomial time~\cite{ILG87,Fed92,Rot92,Fed94} (see the final paragraph of Section~\ref{se:preliminaries} for more detail on the role played by each of these references).

Since finding a weight-minimal stable matching in {\sc sr} instances is an $\NP$-hard task~\cite{Fed92}, it follows that solving the problem with forced and forbidden edges requires different methods from the aforementioned weighted transformation. Fleiner et al.~\cite{FIM07} showed that any {\sc sr} instance with forbidden edges can be converted into another stable matching problem involving ties that can be solved in $O(m)$ time~\cite{IM02} and the transformation has the same time complexity as well. Forced edges can easily be eliminated by forbidding all edges adjacent to them, therefore we can state the following result.

\begin{theorem}[Fleiner et al.~\cite{FIM07}]  
\label{th:sr_decision}
The problem of finding a stable matching in an {\sc sr}  instance with forced and forbidden edges or reporting that none exists is solvable in $O(m)$ time.
\end{theorem}

As we have seen so far, answering the question as to whether a stable solution containing all forced and avoiding all forbidden edges exists can be solved efficiently in the case of both {\sc sm} and {\sc sr}. We thus concentrate on cases where the answer to this question is no. What kind of approximate solutions exist then and how can we find them?

\paragraph{Our contribution} Since optimality is defined by two criteria, it is straightforward to define approximation from those two points of view. In case~BP, all constraints on restricted edges must be satisfied, and we seek a matching with the minimum number of blocking edges. In case~CV, we seek a stable matching that violates the fewest constraints on restricted edges.  The optimisation problems that arise from each of these cases are defined formally in Section~\ref{se:preliminaries}.

In Section~\ref{se:almost_stable}, we consider case~BP: that is, all constraints on restricted edges must be fulfilled, while the number of blocking edges is minimised. We show that in the {\sc sm} case, this problem is computationally hard and not approximable within $n^{1-\varepsilon}$ for any $\varepsilon > 0$, unless $\P = \NP$. We also discuss special cases for which this problem becomes tractable. This occurs if the maximum degree of the graph is at most~2 or if the number of blocking edges in the optimal solution is a constant. We point out a striking difference in the complexity of the two cases with only forbidden and only forced edges: the problem is polynomially solvable if the number of forbidden edges is a constant, but by contrast it is $\NP$-hard even if the instance contains a single forced edge. We also prove that when the restricted edges are either all forced or all forbidden, the optimisation problem remains $\NP$-hard even on very sparse instances, where the maximum degree of a vertex is~3.

Case~CV, where the number of violated constraints on restricted edges is minimised while stability is preserved, is studied in Section~\ref{se:viol_const}. It is a rather straightforward observation that in {\sc sm}, the setting can be modelled and efficiently solved with the help of edge weights. Here we show that on non-bipartite graphs, the problem becomes $\NP$-hard
. As in case~BP, we also discuss the complexity of degree-constrained restrictions and establish that the $\NP$-hardness results remain intact even for graphs with degree at most~3, while the case with degree at most~2 is polynomially solvable.

A structured overview of our results for general {\sc sm} and {\sc sr} instances is contained in Table~\ref{ta:results}. 
\begin{table}[ht]
	\centering
    \resizebox{\columnwidth}{!}{
		\begin{tabular}{|c|c|c|}
		\hline
			\phantom{n}& Stable Marriage & Stable Roommates \\ \hline
			\begin{tabular}{c}case BP: \\ min \# blocking edges \end{tabular} & \begin{tabular}{c}$\NP$-hard to approximate \\ within $n^{1-\varepsilon}$ \end{tabular}&  \begin{tabular}{c}$\NP$-hard to approximate \\ within $n^{1-\varepsilon}$ \end{tabular}\\  \hline
			 \begin{tabular}{c}case CV: min \# violated \\ restricted edge constraints \end{tabular} &\begin{tabular}{c}solvable \\ in polynomial time \end{tabular} & \begin{tabular}{c}$\NP$-hard
             \end{tabular}\\
		\hline
		\end{tabular}}\newline
\caption{Summary of results}
\label{ta:results}
\end{table}

\section{Preliminaries and techniques}
\label{se:preliminaries}

In this section, we introduce the notation used in the remainder of the paper and also define the key problems that we investigate later. An instance $\mathcal{I} =(G, O)$ of the \emph{Stable Marriage problem} ({\sc sm}) consists of a bipartite graph $G =(U \cup W, E)$ with $n$ vertices and $m$ edges, and a set $O$: the set of strictly ordered, but not necessarily complete preference lists. These lists are provided on the set of adjacent vertices at each vertex. The \emph{Stable Roommates problem} ({\sc sr}) differs from {\sc sm} in one sense: the underlying graph $G$ need not be bipartite. In both {\sc sm} and {\sc sr}, a matching $M$ in $G$ is sought, assigning each agent to at most one partner. If a vertex $v\in V(G)$ is matched in $M$, we denote by $M(v)$ the partner of $v$ in $M$.  An edge $uw \in E \setminus M$ \emph{blocks} $M$, or forms a \emph{blocking pair} of $M$ if either $u$ is unmatched or prefers $w$ to $M(u)$, and either $w$ is unmatched or prefers $u$ to $M(w)$. A matching that is not blocked by any edge is called \emph{stable}.

As already mentioned in the Introduction, an {\sc sr} instance need not admit a stable solution. The number of blocking edges is a characteristic property of every matching. The set of edges blocking $M$ is denoted by~$bp(M)$. A natural goal is to find a matching minimising~$|bp(M)|$;
following the consensus in the literature, such a matching is called \emph{almost stable}. This approach has a broad literature: almost stable matchings have been investigated in {\sc sm}~\cite{KMV94,HIM09,BMM10} and {\sc sr}~\cite{ABM06,BMM12} instances.

All problems investigated in this paper deal with at least one set of restricted edges. The set of forbidden edges is denoted by~$P$, while $Q$ stands for the set of forced edges. We assume throughout the paper that $P \cap Q = \emptyset$. A matching $M$ satisfies all constraints on restricted edges if $M \cap P = \emptyset$ and $Q \subseteq M$.

In Figure~\ref{fi:base}, a sample {\sc sm} instance on four men and four women can be seen. The preference ordering is shown on the edges. For instance, vertex $u_2$ ranks $w_1$ best, then $w_4$, and $w_2$ last. The set of forbidden edges $P = \{ u_2w_2, u_3w_3\}$ is marked by dotted grey edges. The unique stable matching $M = \{u_1w_1, u_2w_2, \linebreak[0] u_3w_3, u_4w_4\}$ contains both forbidden edges. Later on, we will return to this sample instance to demonstrate approximation concepts on it.

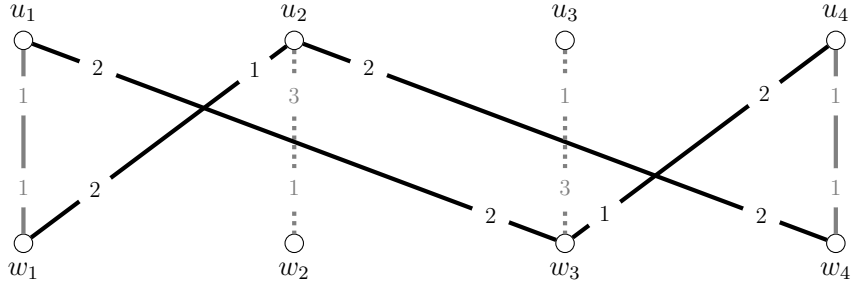
\begin{figure}[h]
\centering
\begin{tikzpicture}[scale=0.9, transform shape]
\node[vertex, label=above:$u_1$] (u_1) at (0, 3) {};
\node[vertex, label=above:$u_2$] (u_2) at (4, 3) {$$};
\node[vertex, label=above:$u_3$] (u_3) at (8, 3) {$$};
\node[vertex, label=above:$u_4$] (u_4) at (12, 3) {$$};

\node[vertex, label=below:$w_1$] (w_1) at (0, 0) {};
\node[vertex, label=below:$w_2$] (w_2) at (4, 0) {};
\node[vertex, label=below:$w_3$] (w_3) at (8,0) {};
\node[vertex, label=below:$w_4$] (w_4) at (12,0) {};

\draw [ultra thick, dotted, gray] (u_2) -- node[edgelabel, near start] {3} node[edgelabel, near end] {1} (w_2);
\draw [ultra thick] (u_2) -- node[edgelabel, very near start] {2} node[edgelabel, very near end] {2}  (w_4);
\draw [ultra thick, dotted, gray] (u_3) -- node[edgelabel, near start] {1} node[edgelabel, near end] {3} (w_3);
\draw [ultra thick, gray] (u_4) -- node[edgelabel, near start] {1} node[edgelabel, near end] {1} (w_4);
\draw [ultra thick, gray] (u_1) -- node[edgelabel, near start] {1} node[edgelabel, near end] {1} (w_1);
\draw [ultra thick] (u_4) -- node[edgelabel, near start] {2} node[edgelabel, very near end] {1} (w_3);
\draw [ultra thick] (u_1) -- node[edgelabel, very near start] {2} node[edgelabel, very near end] {2} (w_3);
\draw [ultra thick] (u_2) --  node[edgelabel, very near start] {1} node[edgelabel, near end] {2}(w_1);

\end{tikzpicture}
\caption{A sample stable marriage instance with forbidden edges}
\label{fi:base}
\end{figure}

The first approximation concept (case~BP described in Section~\ref{se:intro}) is to seek a matching $M$ that satisfies all constraints on restricted edges, but among these matchings, it admits the  minimum number of blocking edges. This leads to the following problem definition.

\begin{pr}{\sc min bp sr restricted}\ \\
	Input: $\mathcal{I} = (G, O, P, Q)$ comprising an {\sc sr} instance $(G,O)$, a set of forbidden edges $P$ and a set of forced edges~$Q$. \\
	Output: A matching $M$ such that $M \cap P = \emptyset$, $Q \subseteq M$ and~$|bp(M)| \leq |bp(M')|$ for every matching $M'$ in~$G$ satisfying $M'\cap P = \emptyset$, $Q \subseteq M'$.
\end{pr}

Special attention is given to two special cases of {\sc min bp sr restricted}: in {\sc min bp sr forbidden}, $Q = \emptyset$, while in {\sc min bp sr forced}, $P = \emptyset$. Note that an instance of {\sc min bp sr forced} or {\sc min bp sr restricted} can always be transformed into an instance of {\sc min bp sr forbidden} by forbidding all edges that are adjacent to a forced edge. This transformation does not affect the number of blocking edges.

According to the other intuitive approximation concept (case~CV described in Section~\ref{se:intro}), stability constraints need to be fulfilled, while some of the constraints on restricted edges are relaxed. The goal is to find a stable matching that violates as few constraints on restricted edges as possible.

\begin{pr}{\sc sr min restricted violations} \ \\
	Input: $\mathcal{I} = (G, O, P, Q)$ comprising an {\sc sr} instance $(G,O)$, a set of forbidden edges $P$ and a set of forced edges~$Q$. \\
	Output: A stable matching $M$ such that~$|M \cap P| +|Q \setminus M| \leq |M' \cap P| + |Q \setminus M'|$ for every stable matching $M'$ in~$G$.
\end{pr}

Just as in the previous approximation concept (referred to as case~BP in Section~\ref{se:intro}), we separate the two subcases with only forbidden and only forced edges. If $Q = \emptyset$, {\sc sr min restricted violations} is referred as {\sc sr min forbidden}, while if $P = \emptyset$, the problem becomes {\sc sr max forced}.  If $P=\emptyset$ or $Q=\emptyset$ then that set is omitted from an instance of {\sc min bp sr restricted} or {\sc sr min restricted violations} as appropriate.

When considering the decision versions of the problems defined in this section, we append {\sc dec} to the problem name and add a positive integer $K$ to the problem instance.  The problem is then to decide whether a feasible solution exists with measure at most $K$.  For example, in the case of the optimisation problem {\sc min bp sr forbidden}, an instance of the decision problem {\sc min bp sr forbidden dec} comprises a tuple $(G,O,P,K)$, where $(G,O,P)$ is as per the definition of {\sc min bp sr forbidden} and $K$ is a positive integer.  The question is whether there is a matching $M$ such that $|M\cap P|=\emptyset$ and $|bp(M)|\leq K$. 

In all discussed problems, $n$ is the number of vertices and $m$ is the number of edges in the graph underlying the particular problem instance. When considering the restriction of any of the above problems to the case of a bipartite graph {\sc sr} is replaced by {\sc sm} in the problem name.

In case~BP, the subcase with only forced edges can be transformed into the other subcase, simply by forbidding edges adjacent to forced edges. This straightforward transformation is not valid for case~CV. Suppose a forced edge was replaced by an unrestricted edge, but all of its adjacent edges were forbidden. A solution that does not contain the original forced edge might contain two of the forbidden edges, violating more constraints than the original solution. Yet most of our proofs are presented for the problem with only forbidden edges, and they require only slight modifications for the case with forced edges.

A powerful tool used in several proofs in our paper is to convert some of these problems into a weighted {\sc sm} or {\sc sr} problem, where the goal is to find a stable matching with the lowest total edge weight, taken over all stable matchings. Irving et al.~\cite{ILG87} were the first to show that weighted {\sc sm} can be solved in polynomial time, giving an $O(n^4 \log n)$ algorithm if the weight function is monotone in the preference ordering, non-negative and integral. Feder~\cite{Fed92,Fed94} showed a method to drop the monotonicity requirement. He also presented the best known bound for the running time of an algorithm for finding a minimum weight stable matching in {\sc sm}: $O(n^{2} \cdot \log({\frac{K}{n^2}}+2)\cdot \min{\{n, \sqrt{K}\}})$, where $K$ is the weight of an optimal solution. Redesigning the weight function to avoid the monotonicity requirement using Feder's method can radically increase~$K$. Fortunately, linear programming techniques allow the conditions to be dropped while retaining polynomial-time solvability. A simple and elegant formulation of the {\sc sm} polytope is known~\cite{Rot92} and using this, a minimum weight stable matching can be computed for all real-valued weight functions in polynomial time via linear programming. For weighted {\sc sr}, finding an optimal matching is $\NP$-hard, but 2-approximable with combinatorial methods, under the assumption of monotone, non-negative and integral weights~\cite{Fed92}. With the help of LP methods, a 2-approximation can be found for every non-negative weight function that satisfies a special monotonicity constraint~\cite{TS97,TS98}.

\section{Almost stable matchings with restricted edges}
\label{se:almost_stable}

In this section, constraints on restricted edges must be fulfilled strictly, while the number of blocking edges is minimised. Our results are presented in three subsections, and most of the results are given for {\sc min bp sm restricted}. Firstly, in Section~\ref{se:gencomp}, basic complexity results are discussed. In particular, we prove that the studied problem {\sc min bp sm restricted} is in general $\NP$-hard and very difficult to approximate. Thus, restricted cases are analysed in Section~\ref{se:bounded_par}. First we assume that the number of forbidden, forced or blocking edges can be considered as a constant. Due to this assumption, two of the three problems that naturally follow from imposing these restrictions become tractable, but surprisingly, not all of them. Then, degree-constrained cases are discussed. We show that the $\NP$-hardness result for {\sc min bp sm restricted} holds even for instances where each preference list is of length at most~3, while on graphs with maximum degree~2, the problem becomes tractable. Finally, in Section~\ref{se:sr} we consider the problem {\sc min bp sr restricted} and briefly elaborate on whether the results established for the bipartite case carry over to the {\sc sr} case.

\subsection{General complexity and approximability results}
\label{se:gencomp}

When minimising the number of blocking edges, one might think that removing the forbidden edges temporarily and then searching for a stable solution in the remaining instance leads to an optimal solution. Such a matching can only be blocked by forbidden edges, but as the upcoming example demonstrates, optimal solutions are sometimes blocked by unrestricted edges exclusively. In some instances, all almost stable solutions admit only non-forbidden blocking edges. Moreover, a man- or woman-optimal almost stable matching with forbidden edges does not always exist.

Let us recall the {\sc sm} instance in Figure~\ref{fi:base}. In the graph with edge set $E(G) \setminus P$, a unique stable matching exists: $M = \{u_1w_1, u_4w_4 \}$. Matching $M$ is blocked by both forbidden edges in the original instance. On the other hand, matching $M_{1} = \{u_1w_1, u_2w_4, u_4w_3\}$ is blocked by exactly one edge: $bp(M_{1}) = u_4w_4$. Similarly, matching  $M_{2} = \{u_1w_3, u_2w_1, u_4w_4\}$ is blocked only by~$u_1w_1$. Therefore, $M_1$ and $M_2$ are both solutions to {\sc min bp sm forbidden} on this instance. One can easily check that $M_{1}$ and $M_{2}$ are the only matchings with the minimum number of blocking edges. They both are blocked only by unrestricted edges. Moreover, $M_{1}$ is better for $u_1, w_1$ and $w_3$, whereas $M_{2}$ is preferred by $u_2, u_4$ and~$w_4$.

In Theorems~\ref{th:minbpsmforbidden} and \ref{th:inappr_minbp} we present two results demonstrating the $\NP$-hardness and inapproximability of special cases of {\sc min bp sm restricted}.

\begin{theorem}
	\label{th:minbpsmforbidden}
	{\sc min bp sm forbidden dec} and {\sc min bp sm forced dec} are $\NP$-complete.  The result holds even if all preference lists are complete.
\end{theorem}

\begin{proof}
Clearly both problems belong to $\NP$.  We show the $\NP$-hardness of both problems by giving a reduction from the following problem:

\begin{pr} {\sc min bp psmi dec}\ \\
	Input: $\mathcal{I} = (G, O, K)$ comprising an {\sc sm} instance $(G,O)$ and a positive integer~$K$. \\
	Output: A perfect matching $M$ such that~$|bp(M)| \leq K$.
\end{pr}

{\sc min bp psmi} denotes the minimisation version of {\sc min bp psmi dec}, in which we seek a perfect matching with the minimum number of blocking pairs, taken over all perfect matchings in $G$.  {\sc min bp psmi-dec} is $\NP$-complete, and unless $\P = \NP$, {\sc min bp psmi} is not approximable within a factor of $n^{1 - \varepsilon}$, for any~$\varepsilon > 0$~\cite{BMM10}.

We firstly show $\NP$-hardness of {\sc min bp sm forbidden dec} and then indicate how to adapt the proof to show a similar result for {\sc min bp sm forced dec}.  We reduce from {\sc min bp psmi dec} as mentioned above.  Given an instance $\mathcal{I} = (G, O, K)$ of this problem we define an instance $\mathcal{I'} = (G', O', P, K)$ of {\sc min bp sm forbidden dec} as follows.  Let $G=(V,E)$ where $U$ and $W$ are the two colour classes of~$G$.  Let $n=|V|$; then $|U|=|W|=n/2$.  Let $U=\{u_1,u_2,\dots,u_{n/2}\}$
and let $W=\{w_1,w_2,\dots,w_{n/2}\}$.  Add the vertices in $V$ to $G'$.
In addition, $K + 1$ new vertices representing women are added to $G'$. They are denoted by $Y=\{y_1,y_2,\dots,y_{K+1}\}$. Similarly, $K+1$ new men $X=\{x_1.x_2.\dots,x_{K+1}\}$ are added to $G'$. Thus, each colour class of $G'$ consists of $n/2 + K + 1$ vertices. 
	
In $O'$ the preference lists of vertices already in $V(G)$ are structured in three blocks. Each man $u_i$ in the original instance $\mathcal{I}$ keeps his preference list in $O$ at the top of his new list in~$O'$. After these vertices, the entire set of newly-introduced women in $Y$ follows, in arbitrary order. Finally, the rest of the women in $W$, not already in $u_i$'s list follow, in arbitrary order. A similar ordering is used when defining the preference list of each~$w_j$. The original list in $O$ is followed by the vertices in $X$, and then the rest of the men in $U$ follow.

The added newly-added vertices in $X\cup Y$ have different preference orderings. Man $x_i$'s list consists of the women in $W$ in arbitrary order, followed by $y_i$, and then the women in $Y\backslash \{y_i\}$ in arbitrary order.  Similarly $y_j$ ranks all men in $U$ first in arbitrary order, followed by $x_j$, and then the men in $X\backslash \{x_j\}$ in arbitrary order.  The preference lists of the vertices in $\mathcal I'$ are shown in Figure~\ref{fig:preflists}.

\begin{figure}[h]
\begin{center}
\resizebox{\columnwidth}{!}{
\begin{tabular}{ r l l l l}
   $u_i$:& $u_i$'s list in $O$& $y_1,y_2,\dots,y_{K+1}$ & rest of women in $W$ & $(1\leq i\leq n/2)$\\
   $w_j$:& $w_j$'s list in $O$& $x_1,x_2,\dots,x_{K+1}$ & rest of men in $U$ & $(1\leq j\leq n/2)$\\
   $x_i$ :& $w_1, w_2, ..., w_{n/2}$ & $y_i$ & women in $Y\backslash \{y_i\}$ & ($1\leq i\leq K+1$)\\
   $y_j$ :& $u_1, u_2, ..., u_{n/2}$ & $x_j$ &  men in $X\backslash \{x_j\}$ &  ($1\leq j\leq K+1$)
 \end{tabular}
 }
\end{center}
\caption{Preference lists in the constructed instance of {\sc min bp sm forbidden dec}.}
\label{fig:preflists}
\end{figure}

Having described $G'$ and $O'$ completely, all that remains is to specify the set of forbidden edges~$P$. Each man $u_i$ has $K+1$ forbidden edges adjacent to him, namely, all edges to the newly-introduced $y_1,y_2,\dots,y_{K+1}$ vertices. Similarly, edges between every $w_j$ and all $x_1,x_2,\dots,x_{K+1}$ vertices are also forbidden. In total, $\mathcal{I'}$ has $n(K+1)$ forbidden edges.
	
	\begin{claim}
If $M$ is a perfect matching in $\mathcal{I}$ admitting at most $K$ blocking edges, then there is a matching $M'$ in $\mathcal{I'}$ with $M' \cap P = \emptyset$ admitting also at most $K$ blocking edges.
	\end{claim}
	
		The construction of $M'$ begins with copying $M$ to~$G'$. Since $M$ is a perfect matching, all vertices in $V(G)$ are matched to vertices in $V(G)$ and thus, no forbidden edge can be in~$M'$. The remaining vertices $x_1,x_2,\dots,x_{K+1}$ and $y_1,y_2,\dots,y_{K+1}$ are paired with each other: each $x_iy_i$ is added to~$M'$.
		
		$M'$ is a perfect matching in $G'$, not containing any of the forbidden edges. Next, we show that no edge in $E(G') \setminus M'$ blocks~$M'$ that did not block $M$ already. First of all, the forbidden edges do not block $M'$, because the preference lists of the vertices already in $V(G)$ were constructed in such a way that the vertices on preference lists in $O$ are better than the vertices in $X\cup Y$, and all $u_i, w_j$ vertices were matched in the perfect matching~$M$. The first $n/2$ choices of any newly-added vertex in $X\cup Y$ are thus not blocking edges. At the same time, all these new vertices are matched to their first-choice partners among the newly-added vertices. Therefore no edge incident to them can block~$M'$. All that remains is to observe that $u_i w_j$ edges blocking $M'$ in $\mathcal{I'}$ already blocked $M$ in~$\mathcal{I}$, because $M$ is the restriction of $M'$ to~$G$. Therefore, the edges blocking $M$ and $M'$ are identical.
	
	\begin{claim}
		If $M'$ is a matching in $\mathcal{I'}$ with $M' \cap P = \emptyset$ admitting at most $K$ blocking edges, then its restriction to $G$ is a perfect matching $M$ in $\mathcal{I}$ admitting at most $K$ blocking edges.
	\end{claim}
	
		First, we discuss some essential structural properties of~$M'$. The forbidden edges are not in $M'$, and at most $K$ of them can block it. Suppose that there is a man $u_i$ not married to any woman $w_j$ in matching $M$. Since $w_j$ ranks exactly $K + 1$ forbidden edges after its listed partners in~$G$, and forbidden edges are the first $n/2$ choices of their other end vertex, all $K + 1$ of them block~$M'$, regardless of the remaining edges in~$M'$. Having derived a contradiction to our assumption that at most $K$ edges block $M'$ in total, we can state that each man $u_i$ is matched in $M'$ to a vertex $w_j$ in $O$. Thus, the restriction of $M'$ to $G$ is a perfect matching with at most $K$ blocking edges.
		
		$\NP$-hardness can be obtained for {\sc min bp sm forced dec} by simply forcing all edges of the form $x_i y_i$ in the above reduction.
\end{proof}

We now strengthen Theorem~\ref{th:minbpsmforbidden} by giving strong lower bounds for the approximability of {\sc min bp sm forbidden} and {\sc min bp sm forced}.  The reduction given in the proof of the next theorem builds on the reduction given in the proof of Theorem~\ref{th:minbpsmforbidden}.
	
\begin{theorem}
\label{th:inappr_minbp}
Each of {\sc min bp sm forbidden} and {\sc min bp sm forced} is not approximable within a factor of $n^{1-\varepsilon}$, for any~$\varepsilon > 0$, unless $\P = \NP$.  The result holds even if all preference lists are complete.
\end{theorem}

\begin{proof}
We will give a reduction from the following $\NP$-complete problem: 

\begin{pr}{\sc exact maximal matching}\ \\
	Input: $\mathcal{I} = (G, K)$ comprising a bipartite graph $G$ and a positive integer~$K$.\\
	Question: Is there a maximal matching $M$ in $G$ such that $|M| = K$?
\end{pr}

{\sc exact maximal matching} is $\NP$-complete even for graphs where all vertices representing men have degree two, while all vertices of the other colour class have degree three~\cite{OMa07}. We show that if there were a polynomial approximation algorithm within a factor of $n^{1-\varepsilon}$ to {\sc min bp sm forbidden}, then it would also find an exact maximal matching in~$\mathcal{I}$.

In our proof, every instance $\mathcal{I}=(G,K)$ of {\sc exact maximal matching} is transformed into an instance $\mathcal{I''}=(G',O',P)$ of {\sc min bp sm forbidden}.  We later show how to adapt the proof for {\sc min bp sm forced}.  Let $n_1$ and $n_2$ denote the size of each colour class in~$\mathcal{I}$, such that $m = 2 n_1 = 3 n_2$.

We show that if there were a polynomial approximation algorithm with a performance guarantee of $n^{1-\varepsilon}$ for {\sc min bp sm forbidden} (where $n$ is the number of vertices in $G'$), then it would solve {\sc exact maximal matching} in polynomial time.  To do so, another transformation is used, involving $\mathcal{I'}$, an instance of {\sc min bp psmi}. In~\cite{BMM10}, an instance $\mathcal{I'}=(G,O)$ of {\sc min bp psmi} is created from $\mathcal{I}$ with special properties. One of them is that if $G$ has a maximal matching of cardinality~$K$, then $\mathcal{I'}$ has a perfect matching admitting exactly $n_1 + n_2$ blocking edges. Otherwise, if $G$ has no maximal matching of cardinality~$K$, then any perfect matching in $\mathcal{I'}$ is blocked by at least $n_1 + n_2 + C$ edges, where $C$ is a huge number. To be more precise, let $B = \left\lceil\frac{3}{\varepsilon}\right\rceil$ and $C = (n_1 + n_2)^{B+1} + 1$ (in~\cite{BMM10}, the value of $B$ was the same to that used here, but the value of $C$ was slightly different). The number of vertices in each colour class of $\mathcal{I'}$ is $3 n_1 + 2mC + 4n_2 -K$.
	
Now we describe how $\mathcal{I'}$ is transformed into~$\mathcal{I''}$. Note that this method is very similar to the one we used in the proof of Theorem~\ref{th:minbpsmforbidden}. Denote by $U$ and $W$ the set of men and women in $\mathcal I'$, and let $u_i$ and $w_j$ denote an arbitrary man and woman in $\mathcal I'$ respectively.  Add these vertices to $\mathcal I''$ and then introduce $C$ new men, namely $X=\{x_1,x_2,\dots,x_C\}$, and $C$ new women, namely $Y=\{y_1,y_2,\dots,y_C\}$. Then each colour class in $\mathcal I''$ consists of $3 n_1 + 2mC + 4n_2 -K + C$ vertices. The preference lists of the vertices in $\mathcal I''$ are shown in Figure~\ref{fig:preflists2}.

\begin{figure}[h]
\begin{center}
\begin{tabular}{rllll}
   $u_i$:& $u_i$'s list in $O$ & $y_1,y_2,\dots,y_C$ & rest of women in $W$ & ($1\leq i\leq |U|$)\\
   $w_j$:& $w_j$'s list in $O$ & $x_1,x_2,\dots,x_C$ & rest of men in $U$ & ($1\leq j\leq |W|$) \\
   $x_i$:& all women in $W$ & $y_i$ & women in $Y\backslash \{y_i\}$ & ($1\leq i\leq C$) \\
   $y_j$:& all men in $U$ & $x_j$ & men in $X\backslash \{x_j\}$ & ($1\leq j\leq C$)
\end{tabular}
\end{center}
\caption{Preference lists in the constructed instance of {\sc min bp sm restricted}.}
\label{fig:preflists2}
\end{figure}

The set of forbidden edges comprises all edges of the form $u_ix_j$ or $w_jy_i$. For {\sc min bp sm forced}, the set of forced edges consists of all edges of the form~$x_i y_i$. Due to this construction, and as in the proof of Theorem~\ref{th:minbpsmforbidden}, if $M$ is a matching in $\mathcal{I''}$ in which there is a man $u_i$ not matched to a woman in $O$, then $M$ is blocked by at least $C$ edges. 

It follows that if $G$ has a maximal matching of size at most $K$ then $\mathcal I''$ has a matching with at most $n_1+n_2$ blocking pairs.  On the other hand if $G$ has no maximal matching of size at most $K$ then any matching in $\mathcal I''$ has at least $C>(n_1+n_2)^{B+1}$ blocking pairs.  Hence an $(n_1+n_2)^B$-approximation algorithm for {\sc min bp sm restricted} or {\sc min bp sm forced} could be used to solve {\sc exact maximal matching} in polynomial time.

To complete the proof it remains to show that if $n$ is the number of vertices in $\mathcal I''$, then $n^{1-\varepsilon}\leq (n_1+n_2)^B$. 
Using Inequalities~\ref{eq:2}-\ref{eq:7} we give an upper bound for~$n$, whilst with Inequalities~\ref{eq:9}-\ref{eq:12b} we establish a lower bound. Then, combining these two in Inequalities~\ref{eq:14}-\ref{eq:17}, we derive that $n^{1-\varepsilon}\leq (n_1+n_2)^B$. Explanations for the steps are given as necessary after each of the three sets of inequalities.
\begin{align}
\label{eq:2}    n &= 2(3 n_1 + 2mC + 4n_2 -K + C) \\
\label{eq:2a} & =6 n_1 + 8n_1C + 8n_2 -2K + 2C \\ 
\label{eq:3}   &\leq 6 n_1 + 8n_1 ((n_1 + n_2)^{B+1} + 1) + 8n_2 - 2K + 2 (n_1 + n_2)^{B+1} + 2 \\
\label{eq:4}  &\leq 14 n_1 + (n_1 + n_2)^{B+1} (8 n_1 + 2) + 8n_2 + 2 \\
\label{eq:5}  &\leq 14 n_1 + 14 n_2 + (n_1 + n_2)^{B+1} (14 n_1 + 14n_2) \\
\label{eq:6} 	&\leq  (n_1 + n_2)^{B+2}  + 14(n_1 + n_2)^{B+2} \\
\label{eq:7} 	&=  15 (n_1 + n_2)^{B+2}
\end{align}
\noindent
In (\ref{eq:2a}) we use that $m=2n_1$, whilst in (\ref{eq:3}) we use that $C = (n_1 + n_2)^{B+1} + 1$ by definition.  To obtain (\ref{eq:4}) we omit $-2K$ whilst in (\ref{eq:5}) we assume that $n_2 \geq 1$ and increase all coefficients to the highest coefficient of 14.  In (\ref{eq:6}) we assume that $n_1\geq 1$, since $B\geq 3$.
\begin{align}
\label{eq:9}    n &=  6 n_1 + 8n_1C + 9n_2 -2K + 2C \\
\label{eq:11}     &>(n_1 + n_2)^{B+1}  \\
\label{eq:12b}     &\geq 15^B
\end{align}
\noindent
In (\ref{eq:9}) we keep only $C$ from the right-hand side of the equality above and use the fact that $C>(n_1+n_2)^{B+1}$, whilst in (\ref{eq:12b}) we assume without loss of generality that $n_2\geq 6$ so $n_1\geq 9$ (recall that $2n_1=3n_2$).
\begin{align}
\label{eq:14}C &>(n_1 + n_2)^B \\
\label{eq:15} &\geq 15^{- \frac{B}{B+2}} n^{\frac{B}{B+2}}\\
\label{eq:16} &\geq n^{1-\frac{3}{B+2}}\\
\label{eq:17}  &\geq n^{1-\varepsilon}
\end{align}
\noindent
Here (\ref{eq:15}) follows by (\ref{eq:2})-(\ref{eq:7}); (\ref{eq:16}) follows by (\ref{eq:9})-(\ref{eq:12b}) and (\ref{eq:17}) uses the fact that $B \geq \frac{3}{\varepsilon}$.
\end{proof}

\subsection{Bounded parameters}
\label{se:bounded_par}

Our results presented so far show that {\sc min bp sm restricted} is computationally hard even if $P=\emptyset$ or $Q=\emptyset$. Yet if certain parameters of the instance or the solution can be considered as a constant, the problem can be solved in polynomial time. Theorem~\ref{th:minbppconstant} firstly shows that this is true for {\sc min bp sm forbidden}.

\begin{theorem}
\label{th:minbppconstant}
	{\sc min bp sm forbidden} is solvable in $O(m^{L+1})$ time, where $L=|P|$, which is polynomial if $L$ is a constant.
\end{theorem}

	

\begin{proof} We firstly show how to solve {\sc min bp sm forbidden dec} in polynomial time.  We assume that, for the purposes of this proof, the problem definition is modified so that, given an instance $\mathcal{I} = (G, O, P, K)$, we are required to find a matching $M$ in $G$ such that $M\cap P=\emptyset$ and $|bp(M)|\leq K$, or report that no such matching exists.

Our first observation is that this problem is trivially solvable if the target value $K$ satisfies $K \geq L$. In this case, deleting the $L$ forbidden edges from $E(G)$ and finding a stable matching in the remaining graph delivers a matching that is blocked in the original instance by only a subset of the removed edges (if any). Thus, a matching $M$ with $M \cap P = \emptyset$ and $|bp(M)|\leq L\leq K$ always exists.

Now assume that~$K<L$. Suppose firstly that there is a matching $M$ with $M \cap P = \emptyset$ and $|bp(M)| = k \leq K < L$. If those $k$ blocking edges are deleted from~$E(G)$, then $M$ is a stable matching in the remainder of $G$, and $M$ contains none of the forbidden edges. Note that we did not specify which edges block~$M$: they can be both forbidden and unrestricted. 

Hence to solve {\sc min bp sm forbidden dec} we generate all subsets $S$ of potential blocking edges, where $|S|\leq K$.  After deleting the edges in $S$ from $G$, we try to find, in the remaining graph, a matching $M$ such that $M\cap P=\emptyset$ and $M$ is stable, or we report that no such matching exists.  By Theorem~\ref{th:bip_decision}, this step can be accomplished in in $O(m)$ time. If such a matching $M$ exists, then it admits at most $K$ blocking edges in $\mathcal I$.   

Thus $\sum_{i = 0}^{K}{m \choose i}=\sum_{i = 0}^{L}{m \choose i}$ subsets are generated to determine whether the desired matching exists.  The number of rounds is thus $O(m^L)$, while each round takes $O(m)$ time to complete.  The overall running time is $O(m^{L+1})$.

We now show how to use the above approach in order to solve {\sc min bp sm forbidden}.  If we find a solution during the course of this process then $G$ admits a matching $M$ such that $M\cap P=\emptyset$ and $|bp(M)|\leq L$.  In order to minimise $|bp(M)|$ it suffices to use the technique in the previous paragraph in combination with a binary search procedure on values of $K\leq L$.  This requires $O(\log L)$ invocations of the algorithm for the decision problem, which is a constant, and hence the overall time complexity remains $O(m^{L+1})$.
\end{proof}

In sharp contrast to the previous result on polynomial solvability when the number of forbidden edges is small, we state the following theorem for {\sc min bp sm forced dec}.

\begin{theorem}
\label{th:minbpforced1}
{\sc min bp sm forced dec} is $\NP$-complete even if $|Q| = 1$.
\end{theorem}

\begin{proof}
The $\NP$-complete problem we reduce to {\sc min bp sm forced dec} is {\sc exact maximal matching}.  As previously mentioned, this problem is $\NP$-complete even for graphs where all vertices representing men have degree two, while all vertices of the other side have degree three~\cite{OMa07}.  Hence suppose we are given an instance $\mathcal I=(G,K)$ of this restriction, where in $G$, $U_0$ and $W_0$ are the two colour classes and $E$ is the edge set.

In this proof, we construct a {\sc min bp sm forced dec} instance $\mathcal{I'}=(G',O',Q,K')$ with a single forced edge in such a way that there is a maximal matching of cardinality $K$ in $\mathcal{I}$ if and only if there is a matching containing the forced edge and admitting exactly $K'=|U_0|+|W_0|$ blocking edges in~$\mathcal{I'}$. Our construction is based on ideas presented in~\cite{BMM10}.
    
All vertices in $G$ rank their edges in an arbitrary but fixed order. We will refer to these labels when constructing~$\mathcal{I'}$.  We now describe~$\mathcal{I'}$. The vertex set of graph $G'$ in $\mathcal{I'}$  can be partitioned into seven sets: $U$, $V$, $W$, $Z$, $S_1$, $S_2$, $X$ and~$Y$, where $U\cup V\cup X\cup S_1$ are the men and $W\cup Z\cup Y\cup S_2$ are the women. Specific subgraphs of $G'$ are referred to as $u$-gadgets, $w$-gadgets, together with a special gadget containing the forced edge; see Figure~\ref{fi:gadgets_q1}. Aside from these, $G'$ also contains some extra vertices, the so-called \emph{garbage collectors}, partitioned into two sets: $X$ and~$Y$. Later we will see that these garbage collectors are paired to the vertices not covered by the matching in~$G$. To that end, $|X| = |W_0| - K$ and  $|Y| = |U_0| - K$. The whole construction is illustrated in Figure~\ref{fi:minbpsmforcedq1}.
    
\begin{figure}[h]
\centering
\begin{minipage}{0.3\textwidth}
\begin{tikzpicture}[scale=0.9, transform shape]

\tikzstyle{vertex} = [circle, draw=black, scale=0.7]
\tikzstyle{specvertex} = [circle, draw=black, scale=0.7]
\tikzstyle{edgelabel} = [circle, fill=white, scale=0.5, font=\huge]
\pgfmathsetmacro{\d}{2}
\pgfmathsetmacro{\b}{3}


\node[specvertex, label=left:$z_1$] (z1) at (0, 1.5*\d) {};
\node[specvertex, label=left:$z_2$] (z2) at (0, 0.5*\d) {};
\node[specvertex, label=right:$u_1$] (u1) at (\b, 2*\d) {};
\node[specvertex, label=right:$u_2$] (u2) at (\b, \d) {};
\node[specvertex, label=right:$u_3$] (u3) at (\b, 0) {};

\draw [ultra thick] (z1) -- node[edgelabel, near start] {1} node[edgelabel, near end] {1} (u1);
\draw [ultra thick] (z1) -- node[edgelabel, near start] {2} node[edgelabel, near end] {1} (u3);
\draw [ultra thick] (z2) -- node[edgelabel, near start] {1} node[edgelabel, near end] {1} (u2);
\draw [ultra thick] (z2) -- node[edgelabel, near start] {2} node[edgelabel, near end] {2} (u3);

\end{tikzpicture}
\end{minipage}\hspace{2mm}\begin{minipage}{0.3\textwidth}
\begin{tikzpicture}[scale=0.9, transform shape]

\tikzstyle{vertex} = [circle, draw=black, scale=0.7]
\tikzstyle{specvertex} = [circle, draw=black, scale=0.7]
\tikzstyle{edgelabel} = [circle, fill=white, scale=0.5, font=\huge]
\pgfmathsetmacro{\d}{2}
\pgfmathsetmacro{\b}{3}


\node[specvertex, label=left:$w_1$] (w1) at (0, 3*\d) {};
\node[specvertex, label=left:$w_2$] (w2) at (0, 2*\d) {};
\node[specvertex, label=left:$w_3$] (w3) at (0, 1*\d) {};
\node[specvertex, label=left:$w_4$] (w4) at (0, 0*\d) {};

\node[specvertex, label=right:$v_1$] (v1) at (\b, 2.5*\d) {};
\node[specvertex, label=right:$v_2$] (v2) at (\b, 1.5*\d) {};
\node[specvertex, label=right:$v_3$] (v3) at (\b, 0.5*\d) {};

\draw [ultra thick] (w1) -- node[edgelabel, near start] {1} node[edgelabel, near end] {1} (v1);
\draw [ultra thick] (w2) -- node[edgelabel, near start] {1} node[edgelabel, near end] {1} (v2);
\draw [ultra thick] (w3) -- node[edgelabel, near start] {1} node[edgelabel, near end] {1} (v3);
\draw [ultra thick] (w4) -- node[edgelabel, near start] {1} node[edgelabel, near end] {2} (v1);
\draw [ultra thick] (w4) -- node[edgelabel, near start] {2} node[edgelabel, near end] {2} (v2);
\draw [ultra thick] (w4) -- node[edgelabel, near start] {3} node[edgelabel, near end] {2} (v3);
\end{tikzpicture}
\end{minipage}\hspace{4mm}\begin{minipage}{0.3\textwidth}
\begin{tikzpicture}[scale=0.9, transform shape]

\tikzstyle{vertex} = [circle, draw=black, scale=0.7]
\tikzstyle{specvertex} = [circle, draw=black, scale=0.7]
\tikzstyle{edgelabel} = [circle, fill=white, scale=0.5, font=\huge]
\pgfmathsetmacro{\d}{2}
\pgfmathsetmacro{\b}{3}

\node[specvertex, label=right:$u_0''$] (u0'') at (0, \d*2) {};
\node[specvertex, label=right:$u_0'$] (u0') at (0, 0) {};
\node[specvertex, label=right:$u_0$] (u0) at (0, \d) {};

\node[specvertex, label=left:$w_0''$] (w0'') at (-\b, \d*2) {};
\node[specvertex, label=left:$w_0'$] (w0') at (-\b, 0) {};
\node[specvertex, label=left:$w_0$] (w0) at (-\b, \d) {};

\draw [ultra thick] (u0) -- node[edgelabel, near start] {last} node[edgelabel, near end] {last} (w0);
\draw [ultra thick] (u0') -- node[edgelabel, near start] {2} (w0);
\draw [ultra thick] (u0') -- node[edgelabel, near start] {1} (w0');
\draw [ultra thick] (u0) --  node[edgelabel, near end] {last} (w0');

\draw [ultra thick] (u0'') -- node[edgelabel, near end] {1} (w0'');
\draw [ultra thick] (u0) -- node[edgelabel, near end] {2} (w0'');
\draw [ultra thick] (u0'') -- node[edgelabel, near start] {last} (w0);
\end{tikzpicture}
\end{minipage}
\caption{A $u$-gadget, a $w$-gadget and the special gadget}
\label{fi:gadgets_q1}
\end{figure}
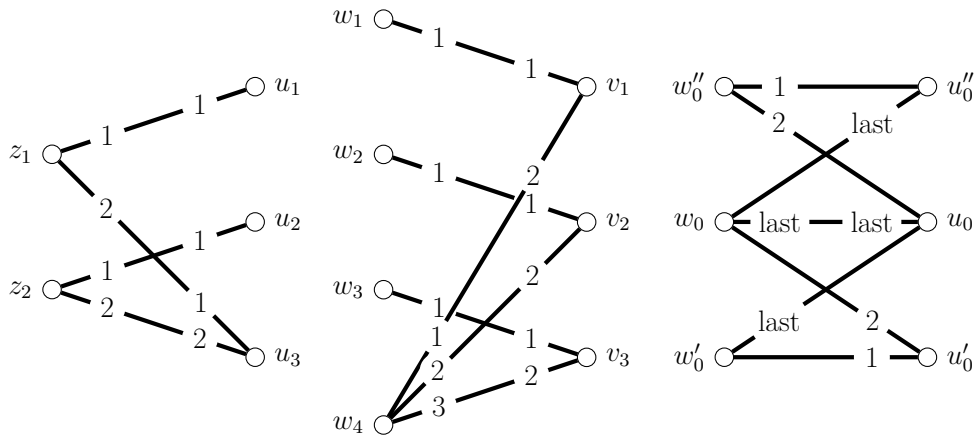

	Each $u$-gadget replaces a vertex $u \in U_0$ in~$G$. It is defined on five vertices: $u_1, u_2, u_3 \in U$ and $z_1, z_2 \in Z$. Its edges and the preferences on them are shown in Figure~\ref{fi:gadgets_q1}. Two \emph{interconnecting edges} connect the special gadget to~$u_3$, and an interconnecting edge connects the special gadget to each of $u_1$ and $u_2$.  These edges are ranked last in the case of $u_1$ and $u_2$, and ranked as the last two edges by $u_3$. It is described later which vertices of the special gadget are incident to these interconnecting edges. The $u$-gadget also has edges to all $w$-gadgets representing vertices in $W_0$ to which $u$ was adjacent. After describing the $w$-gadget, we elaborate on the position of these edges, referred to as \emph{relevant edges}. Aside from these, every $u_1$ has edges to all garbage collectors in~$Y$. These edges are all worse than the relevant edges of~$u_1$ and they are ranked arbitrarily at the bottom of $u_1$'s list. The vertices in $Y$ also rank all $u_1$ vertices arbitrarily.
	
	The $w$-gadgets are structured similarly. Each gadget consists of seven vertices: $w_1, w_2, w_3, w_4 \in W$ and $v_1, v_2, v_3 \in V$. Aside from the edges within the gadget, it has two interconnecting edges between $w_4$ and vertices in the special gadget (described in detail later), and three relevant edges between $w_1, w_2, w_3$ and vertices of $u$-gadgets. These are the edges drawn in accordance with the edge labels. Suppose in $\mathcal I$, edge $uw$ was ranked $i$th by $u$ and $j$th by $w$, where $i \in \left\{1,2 \right\}$ and $j \in \left\{1,2,3 \right\}$. Then in $\mathcal I'$, $u_i$ in the $u$-gadget is connected to $w_j$ in the $w$-gadget. Therefore, each edge in $\mathcal{I}$ is transformed into a single edge in $\mathcal{I'}$ and each $u_i$, $i \in \left\{1,2 \right\}$ and $w_j$, $j \in \left\{1,2,3 \right\}$, has exactly one relevant edge. All of these edges are second choices of both of their end vertices. In addition to these, if $uw \in E(G)$, but the corresponding $u$- and a $w$-gadgets are not yet connected by $u_1w_1$, we add $u_1w_1$, which is referred as an \emph{adjacency edge}. This edge is ranked by both $u_1$ and $w_1$ after their relevant edges, but ahead of their edges to garbage collectors. Similar to $u$-gadgets, $w$-gadgets are also connected to garbage collectors. Each $w_1$ vertex has $|W_0|-K$ edges to the vertices in $X$, ranked arbitrarily at the bottom of $w_1$'s preference list. Also the vertices in $X$ rank the $w_1$ vertices arbitrarily.
	
	The special gadget is defined on only six vertices in the set $S_1\cup S_2$, where $S_1=\{u_0, u_0',u_0''\}$ and $S_2=\{w_0,w_0',w_0''\}$. The unique forced edge in the entire instance is~$u_0 w_0$. Apart from $u_0'$ and $w_0''$, they are connected to $u$- and $w$-gadgets. In each $u$-gadget, $u_3$ is adjacent to $w_0$ and $w_0'$, and each of $u_1$ and $u_2$ is adjacent to $w_0$.  In each $w$-gadget, $w_4$ is adjacent to $u_0$ and~$u_0''$, and each of $w_1$, $w_2$ and $w_3$ is adjacent to $u_0$.  Moreover, $u_0$ and $w_0$ are connected to all garbage collectors of the opposite side via additional interconnecting edges.  These edges are ranked last by the vertices in $X$ and~$Y$.  The four vertices $u_0,u_0'',w_0,w_0'$ prefer their interconnecting edges to their edges inside of the special gadget. 

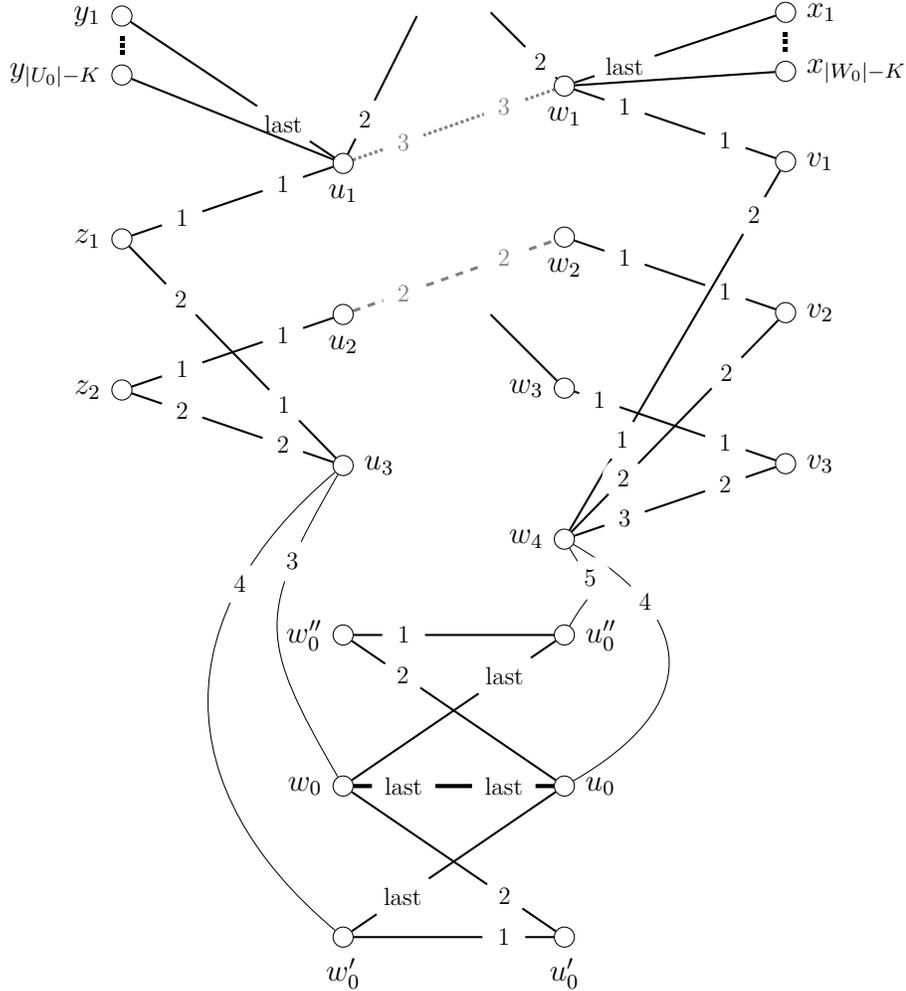
\begin{figure}[htb]
\centering
\resizebox{\columnwidth}{!}{
\begin{minipage}{1\textwidth}
\pgfmathsetmacro{\d}{2.05}
\pgfmathsetmacro{\b}{3}
\pgfmathsetmacro{\e}{0.6}

\def\ugadget#1{
\begin{scope}[shift={#1}]
\node[vertex, label=left:$z_1$] (z1) at (0, 1.5*\d) {};
\node[vertex, label=left:$z_2$] (z2) at (0, 0.5*\d) {};
\node[vertex, label=below:$u_1$] (u1) at (\b, 2*\d) {};
\node[vertex, label=below:$u_2$] (u2) at (\b, \d) {};
\node[vertex, label=right:$u_3$] (u3) at (\b, 0) {};

\draw [thick] (z1) -- node[edgelabel, near start] {1} node[edgelabel, near end] {1} (u1);
\draw [thick] (z1) -- node[edgelabel, near start] {2} node[edgelabel, near end] {1} (u3);
\draw [thick] (z2) -- node[edgelabel, near start] {1} node[edgelabel, near end] {1} (u2);
\draw [thick] (z2) -- node[edgelabel, near start] {2} node[edgelabel, near end] {2} (u3);
\end{scope}
}

\def\wgadget#1{
\begin{scope}[shift={#1}]
\node[vertex, label=below:$w_1$] (w1) at (0, 3*\d) {};
\node[vertex, label=below:$w_2$] (w2) at (0, 2*\d) {};
\node[vertex, label=left:$w_3$] (w3) at (0, 1*\d) {};
\node[vertex, label=left:$w_4$] (w4) at (0, 0*\d) {};

\node[vertex, label=right:$v_1$] (v1) at (\b, 2.5*\d) {};
\node[vertex, label=right:$v_2$] (v2) at (\b, 1.5*\d) {};
\node[vertex, label=right:$v_3$] (v3) at (\b, 0.5*\d) {};

\draw [thick] (w1) -- node[edgelabel, near start] {1} node[edgelabel, near end] {1} (v1);
\draw [thick] (w2) -- node[edgelabel, near start] {1} node[edgelabel, near end] {1} (v2);
\draw [thick] (w3) -- node[edgelabel, very near start] {1} node[edgelabel, near end] {1} (v3);
\draw [thick] (w4) -- node[edgelabel, near start] {1} node[edgelabel, very near end] {2} (v1);
\draw [thick] (w4) -- node[edgelabel, near start] {2} node[edgelabel, near end] {2} (v2);
\draw [thick] (w4) -- node[edgelabel, near start] {3} node[edgelabel, near end] {2} (v3);
\end{scope}
}

\def\specgadget#1{
\begin{scope}[shift={#1}]
\node[vertex, label=right:$u_0''$] (u0'') at (0, \d*2+\e) {};
\node[vertex, label=below:$u_0'$] (u0') at (0, 0+\e) {};
\node[vertex, label=right:$u_0$] (u0) at (0, \d+\e) {};

\node[vertex, label=left:$w_0''$] (w0'') at (-\b, \d*2+\e) {};
\node[vertex, label=below:$w_0'$] (w0') at (-\b, 0+\e) {};
\node[vertex, label=left:$w_0$] (w0) at (-\b, \d+\e) {};

\draw [ultra thick] (u0) -- node[edgelabel, near start] {last} node[edgelabel, near end] {last} (w0);
\draw [thick] (u0') -- node[edgelabel, near start] {2} (w0);
\draw [thick] (u0') -- node[edgelabel, near start] {1} (w0');
\draw [thick] (u0) --  node[edgelabel, near end] {last} (w0');

\draw [thick] (u0'') -- node[edgelabel, near end] {1} (w0'');
\draw [thick] (u0) -- node[edgelabel, near end] {2} (w0'');
\draw [thick] (u0'') -- node[edgelabel, near start] {last} (w0);
\end{scope}
}

\begin{tikzpicture}
\ugadget{(0,1)}
\wgadget{(6,0)}
\specgadget{(6,-6)}
\begin{scope}[on background layer]
\draw [very thick, densely dotted, gray] (u1) -- node[edgelabel, near start] {3} node[edgelabel, near end] {3} (w1);
\draw [very thick, dashed, gray] (u2) -- node[edgelabel, near start] {2} node[edgelabel, near end] {2} (w2);
\end{scope}

\draw [] (u3) to[out=-120,in=120, distance=2.2cm ]  node[edgelabel, near start] {3}  (w0); 
\draw [] (u3) to[out=-140,in=140, distance=3cm ]  node[edgelabel, near start] {4}  (w0');
\draw [] (w4) to[out=-30,in=30, distance=2cm ]  node[edgelabel, near start] {4} (u0);
\draw [] (w4) to[out=-60,in=60, distance=1cm ]  node[edgelabel, near start] {5} (u0'');

\draw [ thick] (u1) -- node[edgelabel, near start]{2}   ($(u1) + (1,2)$);
\draw [ thick] (w1) -- node[edgelabel, near start]{2}   ($(w1) + (-1,1)$);
\draw [ thick] (w3) -- ($(w3) + (-1,1)$);

\node[vertex, label=right:$x_1$] (x1) at ($(w1) + (3,1)$) {};
\node[fill=black, scale=0.2] (x2) at ($(x1) + (0,-0.3)$) {};
\node[fill=black, scale=0.2] (x3) at ($(x2) + (0,-0.1)$) {};
\node[fill=black, scale=0.2] (x4) at ($(x3) + (0,-0.1)$) {};
\node[vertex, label=right:$x_{|W_0|-K}$] (xlast) at ($(w1) + (3,0.2)$) {};
\draw [thick] (w1) -- node[edgelabel, near start]{last}  (x1);
\draw [thick] (w1) -- (xlast);

\node[vertex, label=left:$y_1$] (y1) at ($(u1) + (-3,2)$) {};
\node[fill=black, scale=0.2] (y2) at ($(y1) + (0,-0.3)$) {};
\node[fill=black, scale=0.2] (y3) at ($(y2) + (0,-0.1)$) {};
\node[fill=black, scale=0.2] (y4) at ($(y3) + (0,-0.1)$) {};
\node[vertex, label=left:$y_{|U_0|-K}$] (ylast) at ($(u1) + (-3,1.2)$) {};
\draw [thick] (u1) -- node[edgelabel, near start]{last}  (y1);
\draw [thick] (u1) -- (ylast);


\end{tikzpicture}
\end{minipage}
}
\caption{As the dashed grey relevant edge $u_2w_2$ shows, $u$ and $w$ were connected in $\mathcal{I}$ by an edge ranked second by both of them. The dotted grey edge $u_1w_1$ is an adjacency edge.  Some interconnecting edges have been omitted to avoid clutter.}
\label{fi:minbpsmforcedq1}
\end{figure}
		
	\begin{claim}
		Corresponding to each maximal matching $M$ in $\mathcal{I}$ of cardinality $K$ there is a matching $M'$ in $\mathcal{I'}$ with $u_0w_0 \in M'$ and $|bp(M')| = |U_0| + |W_0|$.
	\end{claim}
		\myproof First, the set of relevant edges in $G'$ corresponding to $M$ is chosen. They cover exactly $K$ of the $|U| = 3 |U_0|$ vertices of $U$, and analogously, exactly $K$ of the $|W| = 4 |W_0|$ vertices in~$W$.
		
		In $u$-gadgets, where either of $u_1$ and $u_2$ has a relevant edge in $M'$, the other vertex in~$U$ is matched to its copy in~$Z$. The remaining two vertices of the gadget are then paired to each other. In the other case, if $u$ was unmatched in $M$, then $\{u_2z_2,u_3 z_1\}\subseteq M'$, and $M'(u_1) \in Y$. Given the set of $u_1$ vertices to pair with the garbage collectors in $Y$, we find any stable matching in this subgraph and add it to~$M'$. Note that this step matches the $|U_0| - K$ $u_1$ vertices to the $|U_0| - K$ garbage collectors in~$Y$.
		
		The strategy is similar for the $w$-gadgets.  Suppose that some $w_j$ is already matched to a vertex in~$U$, because that relevant edge corresponds to a matching edge in~$M$. In $M'$ we then match $w_4$ with $v_j$ and pair the remaining two vertices in $W$ with their partners in~$V$. Otherwise, if $w$ was unmatched in $M$, then in $M'$, $w_1$ is matched to a garbage collector, and $\left\{ w_2v_2, w_3v_3, w_4v_1 \right\} \subseteq M'$. In the subgraph induced by the garbage collectors in $X$ and the $w_1$ vertices corresponding to unmatched $w$ vertices we construct a stable matching and add it to~$M'$. This step matches the $|W_0| - K$ $w_1$ vertices to the $|W_0| - K$ garbage collectors in~$X$.
		
		In the special gadget, $u_0 w_0, u_0'w_0'$ and $u_0''w_0''$ are chosen.
		
		Now we investigate the number of blocking edges incident to at least one vertex in any $u$-gadget. The edges running to garbage collectors cannot block, because $M'$ restricted to that subgraph is a stable matching and $u_1$ vertices not matched to garbage collectors have better relevant edges in~$M'$. Since all $u_3$ vertices are matched to their first or second choices, their edges to the special gadget do not block either. Consider now a relevant edge~$u_i w_j\notin M'$. Since $M$ was a maximal matching, either $u$ or $w$ is matched in $M$.  By construction of $M'$, if $u$ is matched in $M$ then $u_i$ prefers $M'(u_i)$ to $w_j$, whilst if $w$ is matched in $M$ then $w_j$ prefers $M'(w_j)$ to $u_i$. Regarding the adjacency edges, they only block $M'$ if both of their end vertices are matched to garbage collectors. But they both are then unmatched and adjacent in $G$, which contradicts to the fact that $M$ is maximal. The only edges remaining are in the $u$-gadgets. In each $u$-gadget, exactly one edge blocks $M'$: if $u$ was matched to its $i$th ranked edge in $M$, then $u_i z_i$ blocks $M'$, otherwise~$u_1 z_1$ blocks~$M'$.  Therefore, up to this point, we have exactly $|U_0|$ blocking edges.
		
		Analogous arguments prove that among the edges incident to vertices in all $w$-gadgets, $|W_0|$ are blocking. In the previous paragraph we discussed that no relevant or adjacency edge blocks~$M'$. The subgraph induced by the garbage collectors and $w_1$ vertices does not contain any blocking edge, because a stable matching was chosen and the $w_1$ vertices not matched to garbage collectors are all matched in $M'$ to a better vertex. Edges connecting $w_4$ vertices and the special gadget are last choice edges of the matched $w_4$ vertices. In the $w$-gadget, exactly one edge blocks $M'$: if $w$ was matched and therefore $u_iw_j \in M'$, then $w_j v_j$, otherwise~$w_1v_1$.
		
		It is easy to see that in the special gadget, none of the four non-matching edges blocks~$M'$. \myqed

	\begin{claim}
		Corresponding to each matching $M'$ in $\mathcal{I'}$ with $u_0w_0 \in M'$ and $|bp(M')| = |U_0| + |W_0|$ there is a maximal matching $M$ in $\mathcal{I}$ of cardinality~$K$.
	\end{claim}
	\myproof First we show that if $u_0 w_0 \in M'$, then each $u$- and $w$-gadget is adjacent to at least one blocking edge. Since $w_0$ prefers all its edges to $u_0 w_0$, if $u_3$ is not matched in $M'$ to its first or second choice edge then $u_3w_0$ blocks $M'$.  But then if $u_3z_i\in M'$ for $i\in \{1,2\}$, it follows that $u_{3-i}z_{3-i}$ blocks $M'$.  The same argument applies to $u_0$ and~$w_4$. If $M(w_4) \in V$, then $M(w_4)$ has a blocking edge, otherwise $u_0w_4$ blocks~$M'$. Therefore, if $|bp(M')| \leq |U_0| + |W_0|$, then each $u$- and $w$-gadget is incident to exactly one blocking edge.
		
   If $u_3w_0'\in M'$ for some $u$-gadget then each of $u_3w_0$ blocks and $u_0'w_0$ blocks $M'$.  Thus $|bp(M')| \geq |U_0| + |W_0|+1$, a contradiction.  Thus $u_3w_0'\notin M'$ for any $u$-gadget.  By a similar argument we can establish that $u_0''w_4\notin M'$ for any $w$-gadget.  If $u_3$ is unmatched in $M'$ for some $u$-gadget then each of $u_3w_0$ and $u_3w_0'$ blocks $M'$, since $u_3'w_0'\notin M'$ for any $u'$-gadget.  Again that would imply that $|bp(M')| \geq |U_0| + |W_0|+1$, a contradiction.  By a similar argument, $w_4$ is matched in $M'$ for each $w$-gadget.
        
    We have established that, for each $u$-gadget, exactly one of $u_iz_i$ blocks $M'$ for some $i\in \{1,2\}$, and for each $w$ gadget, exactly one of $v_jw_j$ blocks $M'$ for some $j\in \{1,2,3\}$, and these are the only blocking pairs of $M'$ in $\mathcal I'$.  Hence in each $u$-gadget, $u_i$ is matched in $M'$ for $i\in \{1,2\}$, for otherwise $u_iw_0$ blocks $M'$.  Similarly in each $w$-gadget, $w_j$ is matched in $M'$ for $j\in \{1,2,3\}$.  It also follows by a similar argument that each member of $X\cup Y$ is matched in $M'$.  Finally we can observe that $u_0'w_0'\in M'$ for otherwise $u_0'w_0$ blocks $M'$, and $u_0''w_0''\in M'$ for otherwise $u_0w_0''$ blocks $M'$.
    
    Thus in each $u$-gadget, $M'(u_3)=z_i$ for some $i\in \{1,2\}$ whilst $M'(w_4)=v_j$ for some $j\in \{1,2,3\}$.  It follows that in $M'$, $u_i$ is matched either via a relevant edge or adjacency edge, or to a garbage collector in $Y$.  Similarly in $M'$, $w_j$ is matched either via a relevant edge or adjacency edge, or to a garbage collector in $X$.  Meanwhile $u_{3-i}$ is matched in $M'$ to his first-choice vertex in $Z$, whilst the two vertices in $\{w_1,w_2,w_3\}\backslash \{w_j\}$ are matched in $M'$ to their first-choice vertices in $V$.  Define a set of edges $M$ in $G$ as follows:
\[M=\{uw\in E: u_iw_j\in M'\mbox{ for some $i\in \{1,2\}$ and $j\in \{1,2,3\}$}\}.\]
Then $M$ is a matching in $G$ and moreover $|M|=K$, since all $|U_0| - K$ vertices in $Y$ are matched in $M'$ to $u_1$ vertices of various $u$-gadgets, and similarly, all $|W_0| - K$ vertices in $X$ are matched in $M'$ to $w_1$ vertices. That leaves $K$ of the $u$-gadgets that contribute a single relevant edge (or adjacency edge) to~$M'$. All that remains to show is that this matching is maximal. Let us suppose otherwise, i.e., there are two gadgets corresponding to vertices $u$ and $w$ in $G$ such that all their vertices in $U$ and~$W$ in $G'$ are matched to either garbage collectors or to their $z$- or $v$-copies. This is only possible if $u_1$ and $w_1$ are both matched to garbage collectors, but then the adjacency edge $u_1w_1$ blocks~$M'$.
\end{proof}

A counterpart to Theorem~\ref{th:minbppconstant} holds in the case of {\sc min bp sm restricted} if the number of blocking pairs in an optimal solution is a constant.

\begin{theorem}
\label{th:minbpbpconstant}
	{\sc min bp sm restricted} is solvable in $O(m^{L+1})$ time, where $L$ is the minimum number of edges blocking an optimal solution, which is polynomial if $L$ is a constant.
\end{theorem}


\begin{proof} Analogously to the proof of Theorem~\ref{th:minbppconstant}, we show how to solve {\sc min bp sm restricted dec} in polynomial time.  Again, we assume for the purposes of this proof that the problem definition is modified so that, given an instance $\mathcal{I} = (G, O, P, Q, K)$, we are required to find a matching $M$ in $G$ such that $M\cap P=\emptyset$, $Q\subseteq M$ and $|bp(M)|\leq K$, or report that no such matching exists.

Suppose there exists a matching $M$ in $G$ such that $M\cap P=\emptyset$, $Q\subseteq M$ and $|bp(M)|\leq K$.  If we form $G'$ by deleting the edges in $bp(M)$ then $M$ is stable in $G'$.  Hence to find $M$ it suffices to generate all subsets $S$ of edges of size at most $K$ and form a graph $G_S$ by deleting $S$ from $G$.  We can determine in linear time whether $G_S$ admits a stable matching $M'$ that satisfies all constraints on restricted edges~\cite{DFFS03}.  If so, $|bp(M')|\leq K$ in $G$.  There are $O(Km^K)$ sets of edges to remove and checking the existence of a stable matching satisfying constraints on restricted edges can be done in $O(m)$ time.

We now show how to use the above approach to solve {\sc min bp sm restricted}.  For each value of $K$, where $K$ starts from 0 and increases by 1 after each iteration, we execute the algorithm in the previous paragraph, noting that it is sufficient to generate all subsets of size exactly $K$ at each iteration.  We terminate as soon as we find a matching $M$ satisfying the constraints on restricted edges such that $|bp(M)|\leq K$.  This process is bound to halt, since by definition, $\mathcal I$ admits a matching $M$ such that $M\cap P=\emptyset$, $Q\subseteq M$ and $|bp(M)|=L$, so $K\leq L$.  Thus the overall time complexity of this approach is $O((L+1)m^{L+1})=O(m^{L+1})$, which is polynomial if $L$ is a constant.
\end{proof}

Next we study the case of degree-constrained graphs; for most hard {\sc sm} and {\sc sr} problems, it is the most common special case to investigate~\cite{IMO09,HIM09,BMM12}. Here, we show in Theorem~\ref{th:minbp33} that {\sc min bp sm restricted} remains computationally hard even for instances with preference lists of length at most~3. On the other hand, according to Theorem~\ref{th:minbp22}, the problem can be solved in polynomial time when the length of preference lists is bounded by~2. 

\begin{theorem}
\label{th:minbp33}
	{\sc min bp sm forbidden dec} and {\sc min bp sm forced dec} are $\NP$-complete.  The result holds even if each agent's preference list is of length at most 3.
\end{theorem}
\begin{proof} We give a reduction from {\sc (2,2)-e3-sat}, a restriction of {\sc satisfiability}, which may be defined as follows:
\begin{pr}{\sc (2,2)-e3-sat}\ \\
	Input: a Boolean formula $B$ in CNF, in which each clause comprises exactly 3 literals and each variable appears exactly twice unnegated and exactly twice negated. \\
	Question: Is there a satisfying truth assignment for~$B$?
\end{pr}
\noindent {\sc (2,2)-e3-sat} is $\NP$-complete~\cite{BKS03}.  Given an instance $B$ of this problem, let us denote the number of variables by $n_B$ and the number of clauses by $m_B$.

Using the simple transformation described in Section~\ref{se:preliminaries}, any {\sc min bp sm forced} instance can be converted into a {\sc min bp sm forbidden} instance without increasing the preference list lengths.  Hence it is sufficient to investigate {\sc min bp sm forbidden}.

Our goal is to construct an instance $\mathcal{I}=(G,O,P)$ of {\sc min bp sm forbidden} such that $B$ is satisfiable if and only if $\mathcal I$ admits a matching $M$ with $M\cap P=\emptyset$ and $|bp(M)|\leq n_B+m_B$.
	
Our construction combines ideas from two papers. Corresponding to $B$, we introduce a variable gadget and a clause gadget. The first one is a slightly more sophisticated variant of the variable gadget used in Theorem~7 of~\cite{BMM10}, to show $\NP$-hardness of finding a maximum cardinality almost stable matching. Our clause gadget is a simplified version of another clause gadget from Theorem~1 in~\cite{BMM12}. There, the Almost Stable Roommates problem is shown to be $\NP$-hard. Both proofs investigate the case with bounded preference lists.

When constructing instance $\mathcal I$ from the given Boolean formula $B$, we keep track of the order of the three literals in each clause and the order of the two unnegated and two negated appearances of each variable. 
	
\begin{center}
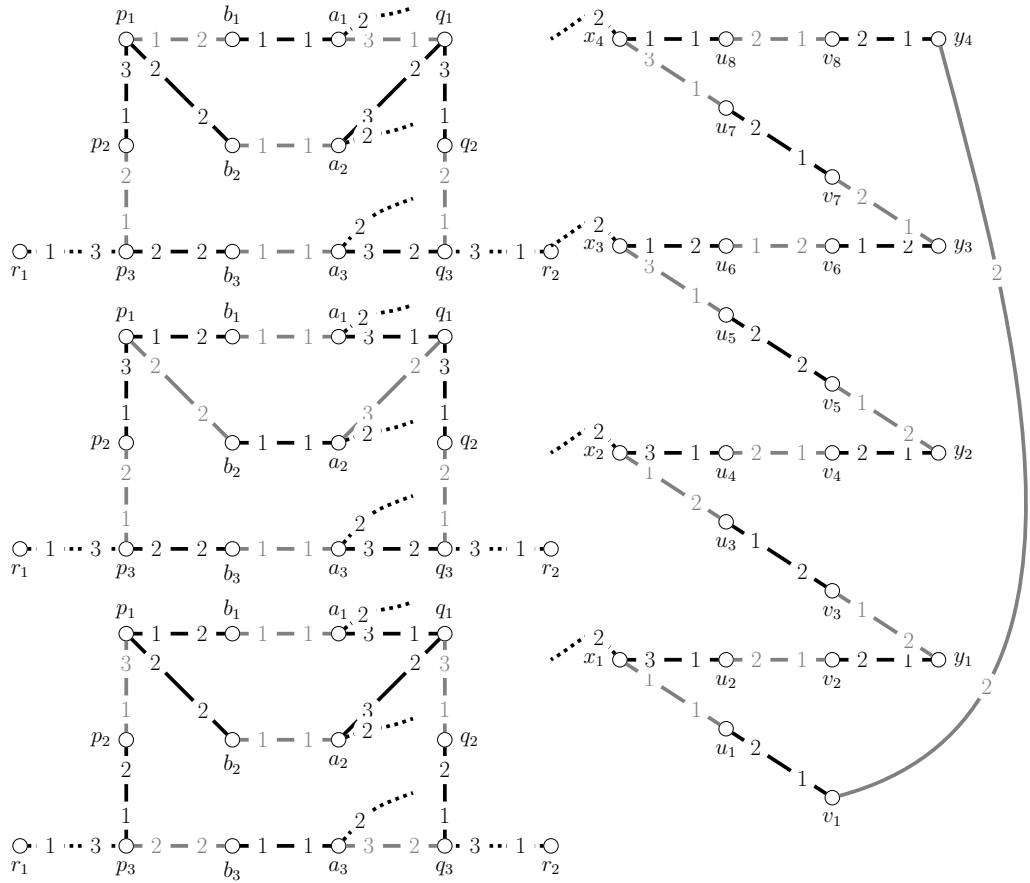
\begin{figure}[h!]
\resizebox{0.9\columnwidth}{!}{
\begin{minipage}{0.6\textwidth}
\begin{tikzpicture}[scale=0.8, transform shape]

\tikzstyle{vertex} = [circle, draw=black, scale=0.7]
\tikzstyle{specvertex} = [circle, draw=black, scale=0.7]
\tikzstyle{edgelabel} = [circle, fill=white, scale=0.5, font=\huge]
\pgfmathsetmacro{\d}{2}


\node[specvertex, label=below:$r_1$] (r1) at (\d*2, 0) {};
\node[specvertex, label=below:$p_3$] (p3) at (\d*3, 0) {};
\node[specvertex, label=below:$b_3$] (b3) at (\d*4, 0) {};
\node[vertex, label=below:$a_3$] (a3) at (\d*5, 0) {};
\node[specvertex, label=below:$q_3$] (q3) at (\d*6, 0) {};
\node[specvertex, label=below:$r_2$] (r2) at (\d*7, 0) {};

\node[specvertex, label=left:$p_2$] (p2) at (\d*3, \d) {};
\node[specvertex, label=below:$b_2$] (b2) at (\d*4, \d) {};
\node[vertex, label=below:$a_2$] (a2) at (\d*5, \d) {};
\node[specvertex, label=right:$q_2$] (q2) at (\d*6, \d) {};

\node[specvertex, label=above:$p_1$] (p1) at (\d*3, \d*2) {};
\node[specvertex, label=above:$b_1$] (b1) at (\d*4, \d*2) {};
\node[vertex, label=above:$a_1$] (a1) at (\d*5, \d*2) {};
\node[specvertex, label=above:$q_1$] (q1) at (\d*6, \d*2) {};

\draw [ultra thick, dotted] (r1) -- node[edgelabel, near start] {1} node[edgelabel, near end] {3} (p3);

\draw [ultra thick, dotted] (r2) -- node[edgelabel, near start] {1} node[edgelabel, near end] {3} (q3);

\draw [ultra thick] (p3) -- node[edgelabel, near start] {2} node[edgelabel, near end] {2} (b3);
\draw [ultra thick, gray] (b3) -- node[edgelabel, near start] {1} node[edgelabel, near end] {1} (a3);
\draw [ultra thick] (a3) -- node[edgelabel, near start] {3} node[edgelabel, near end] {2} (q3);

\draw [ultra thick, gray] (p1) -- node[edgelabel, near start] {1} node[edgelabel, near end] {2} (b1);
\draw [ultra thick] (b1) -- node[edgelabel, near start] {1} node[edgelabel, near end] {1} (a1);
\draw [ultra thick, gray] (a1) -- node[edgelabel, near start] {3} node[edgelabel, near end] {1} (q1); 

\draw [ultra thick] (p1) -- node[edgelabel, near start] {2} node[edgelabel, near end] {2} (b2);
\draw [ultra thick, gray] (b2) -- node[edgelabel, near start] {1} node[edgelabel, near end] {1} (a2);
\draw [ultra thick] (a2) -- node[edgelabel, near start] {3} node[edgelabel, near end] {2} (q1);

\draw [ultra thick] (p1) -- node[edgelabel, near start] {3} node[edgelabel, near end] {1} (p2);
\draw [ultra thick, gray] (p2) -- node[edgelabel, near start] {2} node[edgelabel, near end] {1} (p3);

\draw [ultra thick] (q1) -- node[edgelabel, near start] {3} node[edgelabel, near end] {1} (q2);
\draw [ultra thick, gray] (q2) -- node[edgelabel, near start] {2} node[edgelabel, near end] {1} (q3);

\draw [ultra thick, dotted] (a1) to[out=40,in=-160, distance=0.8cm ]  node[edgelabel, near start] {2} (\d*5.7, \d*2.3);
\draw [ultra thick, dotted] (a2) to[out=20,in=-160, distance=0.8cm ]  node[edgelabel, near start] {2} (\d*5.7, \d*1.2);
\draw [ultra thick, dotted] (a3) to[out=50,in=-160, distance=0.8cm ]  node[edgelabel, near start] {2} (\d*5.7, \d*0.5);

\end{tikzpicture}

\begin{tikzpicture}[scale=0.8, transform shape]

\tikzstyle{vertex} = [circle, draw=black, scale=0.7]
\tikzstyle{specvertex} = [circle, draw=black, scale=0.7]
\tikzstyle{edgelabel} = [circle, fill=white, scale=0.5, font=\huge]
\pgfmathsetmacro{\d}{2}


\node[specvertex, label=below:$r_1$] (r1) at (\d*2, 0) {};
\node[specvertex, label=below:$p_3$] (p3) at (\d*3, 0) {};
\node[specvertex, label=below:$b_3$] (b3) at (\d*4, 0) {};
\node[vertex, label=below:$a_3$] (a3) at (\d*5, 0) {};
\node[specvertex, label=below:$q_3$] (q3) at (\d*6, 0) {};
\node[specvertex, label=below:$r_2$] (r2) at (\d*7, 0) {};

\node[specvertex, label=left:$p_2$] (p2) at (\d*3, \d) {};
\node[specvertex, label=below:$b_2$] (b2) at (\d*4, \d) {};
\node[vertex, label=below:$a_2$] (a2) at (\d*5, \d) {};
\node[specvertex, label=right:$q_2$] (q2) at (\d*6, \d) {};

\node[specvertex, label=above:$p_1$] (p1) at (\d*3, \d*2) {};
\node[specvertex, label=above:$b_1$] (b1) at (\d*4, \d*2) {};
\node[vertex, label=above:$a_1$] (a1) at (\d*5, \d*2) {};
\node[specvertex, label=above:$q_1$] (q1) at (\d*6, \d*2) {};

\draw [ultra thick, dotted] (r1) -- node[edgelabel, near start] {1} node[edgelabel, near end] {3} (p3);
\draw [ultra thick, dotted] (r2) -- node[edgelabel, near start] {1} node[edgelabel, near end] {3} (q3);

\draw [ultra thick] (p3) -- node[edgelabel, near start] {2} node[edgelabel, near end] {2} (b3);
\draw [ultra thick, gray] (b3) -- node[edgelabel, near start] {1} node[edgelabel, near end] {1} (a3);
\draw [ultra thick] (a3) -- node[edgelabel, near start] {3} node[edgelabel, near end] {2} (q3);

\draw [ultra thick] (p1) -- node[edgelabel, near start] {1} node[edgelabel, near end] {2} (b1);
\draw [ultra thick, gray] (b1) -- node[edgelabel, near start] {1} node[edgelabel, near end] {1} (a1);
\draw [ultra thick] (a1) -- node[edgelabel, near start] {3} node[edgelabel, near end] {1} (q1);

\draw [ultra thick, gray] (p1) -- node[edgelabel, near start] {2} node[edgelabel, near end] {2} (b2);
\draw [ultra thick] (b2) -- node[edgelabel, near start] {1} node[edgelabel, near end] {1} (a2);
\draw [ultra thick, gray] (a2) -- node[edgelabel, near start] {3} node[edgelabel, near end] {2} (q1);

\draw [ultra thick] (p1) -- node[edgelabel, near start] {3} node[edgelabel, near end] {1} (p2);
\draw [ultra thick, gray] (p2) -- node[edgelabel, near start] {2} node[edgelabel, near end] {1} (p3);

\draw [ultra thick] (q1) -- node[edgelabel, near start] {3} node[edgelabel, near end] {1} (q2);
\draw [ultra thick, gray] (q2) -- node[edgelabel, near start] {2} node[edgelabel, near end] {1} (q3);

\draw [ultra thick, dotted] (a1) to[out=40,in=-160, distance=0.8cm ]  node[edgelabel, near start] {2} (\d*5.7, \d*2.3);
\draw [ultra thick, dotted] (a2) to[out=20,in=-160, distance=0.8cm ]  node[edgelabel, near start] {2} (\d*5.7, \d*1.2);
\draw [ultra thick, dotted] (a3) to[out=50,in=-160, distance=0.8cm ]  node[edgelabel, near start] {2} (\d*5.7, \d*0.5);

\end{tikzpicture}

\begin{tikzpicture}[scale=0.8, transform shape]

\tikzstyle{vertex} = [circle, draw=black, scale=0.7]
\tikzstyle{specvertex} = [circle, draw=black, scale=0.7]
\tikzstyle{edgelabel} = [circle, fill=white, scale=0.5, font=\huge]
\pgfmathsetmacro{\d}{2}


\node[specvertex, label=below:$r_1$] (r1) at (\d*2, 0) {};
\node[specvertex, label=below:$p_3$] (p3) at (\d*3, 0) {};
\node[specvertex, label=below:$b_3$] (b3) at (\d*4, 0) {};
\node[vertex, label=below:$a_3$] (a3) at (\d*5, 0) {};
\node[specvertex, label=below:$q_3$] (q3) at (\d*6, 0) {};
\node[specvertex, label=below:$r_2$] (r2) at (\d*7, 0) {};

\node[specvertex, label=left:$p_2$] (p2) at (\d*3, \d) {};
\node[specvertex, label=below:$b_2$] (b2) at (\d*4, \d) {};
\node[vertex, label=below:$a_2$] (a2) at (\d*5, \d) {};
\node[specvertex, label=right:$q_2$] (q2) at (\d*6, \d) {};

\node[specvertex, label=above:$p_1$] (p1) at (\d*3, \d*2) {};
\node[specvertex, label=above:$b_1$] (b1) at (\d*4, \d*2) {};
\node[vertex, label=above:$a_1$] (a1) at (\d*5, \d*2) {};
\node[specvertex, label=above:$q_1$] (q1) at (\d*6, \d*2) {};

\draw [ultra thick, dotted] (r1) -- node[edgelabel, near start] {1} node[edgelabel, near end] {3} (p3);
\draw [ultra thick, dotted] (r2) -- node[edgelabel, near start] {1} node[edgelabel, near end] {3} (q3);

\draw [ultra thick, gray] (p3) -- node[edgelabel, near start] {2} node[edgelabel, near end] {2} (b3);
\draw [ultra thick] (b3) -- node[edgelabel, near start] {1} node[edgelabel, near end] {1} (a3);
\draw [ultra thick, gray] (a3) -- node[edgelabel, near start] {3} node[edgelabel, near end] {2} (q3);

\draw [ultra thick] (p1) -- node[edgelabel, near start] {1} node[edgelabel, near end] {2} (b1);
\draw [ultra thick, gray] (b1) -- node[edgelabel, near start] {1} node[edgelabel, near end] {1} (a1);
\draw [ultra thick] (a1) -- node[edgelabel, near start] {3} node[edgelabel, near end] {1} (q1);

\draw [ultra thick] (p1) -- node[edgelabel, near start] {2} node[edgelabel, near end] {2} (b2);
\draw [ultra thick, gray] (b2) -- node[edgelabel, near start] {1} node[edgelabel, near end] {1} (a2);
\draw [ultra thick] (a2) -- node[edgelabel, near start] {3} node[edgelabel, near end] {2} (q1);

\draw [ultra thick, gray] (p1) -- node[edgelabel, near start] {3} node[edgelabel, near end] {1} (p2);
\draw [ultra thick] (p2) -- node[edgelabel, near start] {2} node[edgelabel, near end] {1} (p3);

\draw [ultra thick, gray] (q1) -- node[edgelabel, near start] {3} node[edgelabel, near end] {1} (q2);
\draw [ultra thick] (q2) -- node[edgelabel, near start] {2} node[edgelabel, near end] {1} (q3);

\draw [ultra thick, dotted] (a1) to[out=40,in=-160, distance=0.8cm ]  node[edgelabel, near start] {2} (\d*5.7, \d*2.3);
\draw [ultra thick, dotted] (a2) to[out=20,in=-160, distance=0.8cm ]  node[edgelabel, near start] {2} (\d*5.7, \d*1.2);
\draw [ultra thick, dotted] (a3) to[out=50,in=-160, distance=0.8cm ]  node[edgelabel, near start] {2} (\d*5.7, \d*0.5);

\end{tikzpicture}
\end{minipage}
\begin{minipage}{0.4\textwidth}\vspace{-8mm}
\begin{tikzpicture}[scale=0.8, transform shape]

\tikzstyle{vertex} = [circle, draw=black, scale=0.7]
\tikzstyle{specvertex} = [circle, draw=black, scale=0.7]
\tikzstyle{edgelabel} = [circle, fill=white, scale=0.5, font=\huge]
\pgfmathsetmacro{\d}{1.3}
\pgfmathsetmacro{\b}{2}

\node[vertex, label=left:$x_1$] (x1) at (0, \d*3) {};
\node[vertex, label=left:$x_2$] (x2) at (0, \d*6) {};
\node[vertex, label=left:$x_3$] (x3) at (0, \d*9) {};
\node[vertex, label=left:$x_4$] (x4) at (0, \d*12) {};

\node[vertex, label=below:$u_1$] (u1) at (\b, \d*2) {};
\node[vertex, label=below:$u_2$] (u2) at (\b, \d*3) {};
\node[vertex, label=below:$u_3$] (u3) at (\b, \d*5) {};
\node[vertex, label=below:$u_4$] (u4) at (\b, \d*6) {};
\node[vertex, label=below:$u_5$] (u5) at (\b, \d*8) {};
\node[vertex, label=below:$u_6$] (u6) at (\b, \d*9) {};
\node[vertex, label=below:$u_7$] (u7) at (\b, \d*11) {};
\node[vertex, label=below:$u_8$] (u8) at (\b, \d*12) {};

\node[vertex, label=below:$v_1$] (v1) at (\b*2, \d*1) {};
\node[vertex, label=below:$v_2$] (v2) at (\b*2, \d*3) {};
\node[vertex, label=below:$v_3$] (v3) at (\b*2, \d*4) {};
\node[vertex, label=below:$v_4$] (v4) at (\b*2, \d*6) {};
\node[vertex, label=below:$v_5$] (v5) at (\b*2, \d*7) {};
\node[vertex, label=below:$v_6$] (v6) at (\b*2, \d*9) {};
\node[vertex, label=below:$v_7$] (v7) at (\b*2, \d*10) {};
\node[vertex, label=below:$v_8$] (v8) at (\b*2, \d*12) {};

\node[vertex, label=right:$y_4$] (y4) at (\b*3, \d*12) {};
\node[vertex, label=right:$y_1$] (y1) at (\b*3, \d*3) {};
\node[vertex, label=right:$y_2$] (y2) at (\b*3, \d*6) {};
\node[vertex, label=right:$y_3$] (y3) at (\b*3, \d*9) {};

\draw [ultra thick, gray] (x1) -- node[edgelabel, near start] {1} node[edgelabel, near end] {1} (u1);
\draw [ultra thick, gray] (x2) -- node[edgelabel, near start] {1} node[edgelabel, near end] {2} (u3);
\draw [ultra thick, gray] (x3) -- node[edgelabel, near start] {3} node[edgelabel, near end] {1} (u5);
\draw [ultra thick, gray] (x4) -- node[edgelabel, near start] {3} node[edgelabel, near end] {1} (u7);

\draw [ultra thick] (x1) -- node[edgelabel, near start] {3} node[edgelabel, near end] {1} (u2);
\draw [ultra thick] (x2) -- node[edgelabel, near start] {3} node[edgelabel, near end] {1} (u4);
\draw [ultra thick] (x3) -- node[edgelabel, near start] {1} node[edgelabel, near end] {2} (u6);
\draw [ultra thick] (x4) -- node[edgelabel, near start] {1} node[edgelabel, near end] {1} (u8);

\draw [ultra thick, gray] (u2) -- node[edgelabel, near start] {2} node[edgelabel, near end] {1} (v2);
\draw [ultra thick, gray] (u4) -- node[edgelabel, near start] {2} node[edgelabel, near end] {1} (v4);
\draw [ultra thick, gray] (u6) -- node[edgelabel, near start] {1} node[edgelabel, near end] {2} (v6);
\draw [ultra thick, gray] (u8) -- node[edgelabel, near start] {2} node[edgelabel, near end] {1} (v8);

\draw [ultra thick] (v2) -- node[edgelabel, near start] {2} node[edgelabel, near end] {1} (y1);
\draw [ultra thick] (v4) -- node[edgelabel, near start] {2} node[edgelabel, near end] {1} (y2);
\draw [ultra thick] (v6) -- node[edgelabel, near start] {1} node[edgelabel, near end] {2} (y3);
\draw [ultra thick] (v8) -- node[edgelabel, near start] {2} node[edgelabel, near end] {1} (y4);

\draw [ultra thick, gray] (v3) -- node[edgelabel, near start] {1} node[edgelabel, near end] {2} (y1);
\draw [ultra thick, gray] (v5) -- node[edgelabel, near start] {1} node[edgelabel, near end] {2} (y2);
\draw [ultra thick, gray] (v7) -- node[edgelabel, near start] {2} node[edgelabel, near end] {1} (y3);
\draw [ultra thick, gray] (v1)  to[out=15,in=-75] node[edgelabel, near start] {2} node[edgelabel, near end] {2} (y4);

\draw [ultra thick] (u1) -- node[edgelabel, near start] {2} node[edgelabel, near end] {1} (v1);
\draw [ultra thick] (u3) -- node[edgelabel, near start] {1} node[edgelabel, near end] {2} (v3);
\draw [ultra thick] (u5) -- node[edgelabel, near start] {2} node[edgelabel, near end] {2} (v5);
\draw [ultra thick] (u7) -- node[edgelabel, near start] {2} node[edgelabel, near end] {1} (v7);

\foreach \from in {x1, x2, x3, x4}
   \draw [ultra thick, dotted] (\from) to[out=130,in=30, distance=0.8cm ]  node[edgelabel, near start] {2} ($ (\from) + (-\d,0) $);

\end{tikzpicture}
\end{minipage}
}
\caption{A clause and a variable gadget with their special matchings, marked by grey edges. The dotted edges are forbidden.}
\label{fi:forb33}
\end{figure}
\end{center}

\paragraph{The variable gadget} For each variable in $B$, a \emph{variable gadget}, which is a graph on 44 vertices, is defined. The right hand-side of Figure~\ref{fi:forb33} illustrates the essential part of such a gadget, which is a cycle of length~24. This cycle contains no forbidden edges, and each vertex along it has degree~3 due to being incident to an additional forbidden edge. For sake of simplicity, only four of these forbidden edges are depicted in the figure, namely the ones incident to vertices $x_1, x_2, x_3$ and $x_4$, as they are responsible for the communication between clause and variable gadgets. Each of these four vertices has its third, forbidden edge connected to a clause gadget. These edges are called \emph{interconnecting edges} and are ranked second on the preference list of both of their end vertices.
	
	Consider a variable~$v$ in $B$. Due to the properties of {\sc (2,2)-e3-sat}, $v$ occurs twice in unnegated form, say, in clauses $C_1$ and~$C_2$.  Suppose that $v$'s first unnegated appearance is as the $i$th literal of $C_1$.  This is represented by the interconnecting edge between $x_1$ (in the vertex gadget corresponding to $v$) and vertex $a_i$ in the clause gadget (described below) corresponding to~$C_1$. Similarly, $x_2$ is connected to an $a$-vertex in the clause gadget of~$C_2$. The same variable, $v$, also appears twice in negated form. The relevant variables in the gadgets representing those clauses are connected to $x_3$ and $x_4$. The other end vertices of these two interconnecting edges mark where these two literals appear in their clauses.
	
	As mentioned before, all vertices along the cycle have exactly one forbidden edge attached to them. In the case of $x_1$, $x_2$, $x_3$ and $x_4$, these edges take the form of interconnecting edges.  Regarding the remaining 20 vertices of the cycle, there is a dummy vertex with a forbidden edge attached to each of them, which we call \emph{variable pendant edges} (these edges are not depicted in Figure~\ref{fi:forb33}). This edge is their last choice. These edges guarantee that if any of these 20 vertices remain unmatched in the cycle, then there is a contribution of one blocking edge.
	
	Two special matchings are defined on a variable gadget: $M_T$, denoted by grey edges and $M_F$, comprising the black edges. While $M_T$ is blocked by $x_4u_8$ exclusively, $M_F$ is blocked by $x_1u_1$ exclusively.
	
\begin{claim}
\label{cl:var} 
	Let $M$ be a matching on a variable gadget. If $M$ is not $M_T$ or $M_F$, then it is blocked by at least two edges belonging to the variable gadget.
\end{claim}

\myproof Since $x_1u_1$ and $x_4u_8$ are best-choice edges of both of their end vertices, they block any matching not containing them. If both of them are in $M$, then there is at least one unmatched vertex on the path between $u_8$ and $u_1$ via $y_4$, and another unmatched vertex on the path between $x_4$ and~$x_1$ via $y_2$. The first path comprises vertices with a forbidden edge as their last choice, therefore, it contributes a blocking edge. The only way to avoid additional blocking edges on the second path is to leave either $x_2$ or~$x_3$ unmatched. Since they both are first choices of some other vertex along the cycle, we obtain a second blocking edge.

The remaining case is when exactly one of $x_1u_1$ and $x_4u_8$ is in~$M$. If every second edge in the cycle belongs to~$M$, then it is either $M_T$ or~$M_F$. Otherwise, for simple parity reasons, there are two unmatched vertices on the 24-cycle.  Suppose firstly that $x_1u_1\in M$.  Then $x_4u_8\notin M$ so $x_4u_8$ blocks $M$.  If a $u$-, $v$- or $y$-vertex is unmatched in $M$ then we obtain a further blocking edge from the vertex's variable pendant edge.  Otherwise either $x_2$ is unmatched, so $x_2u_4$ blocks $M$, or $x_3$ is unmatched, so $x_3u_5$ blocks $M$, or $x_4$ is unmatched, so $x_4u_7$ blocks $M$.  The argument is similar if $x_4u_8\in M$. \myqed
	
\paragraph{The clause gadget} To each clause in $B$, a graph on 14 vertices is defined, which we refer to as the \emph{clause gadget}. Three of them, $a_1, a_2$ and $a_3$, are connected to variable gadgets via interconnecting edges, all ranked second. There are three special matchings on a clause gadget, blocked by only a single edge. They can be seen on the left hand-side of Figure~\ref{fi:forb33}. The grey edges selected in each of the clause gadget copies from top to bottom denote $M_1, M_2$ and $M_3$ respectively. 
In addition to the interconnecting edges, a clause gadget has two forbidden edges, namely $p_3r_1$ and $q_3r_2$.
\begin{claim}
\label{cl:clause} 
	Let $M$ be a matching in a clause gadget. If $M$ is not $M_1$, $M_2$ or $M_3$, then it is blocked by at least two edges, both of them belonging to the clause gadget. 
\end{claim}

	\myproof First, suppose that $M \neq M_i$ for all $i \in \{1,2,3\}$ and that $M$ is blocked by at most one edge. Since all three edges connecting $a$ and $b$-vertices are first choices of both of their end vertices, they block any matching not including them. Another restriction arises from the fact that the forbidden edges $r_1p_3$ and $r_2q_3$ ensure that if $p_3$ or $q_3$ is unmatched, they also contribute a blocking edge. Similarly, if $p_1$ or $q_1$ is unmatched, they contribute a blocking edge.
	
	Suppose $b_ia_i \in M$ for all~$i \in \{1,2,3\}$. Then, $p_2$ is matched either to $p_1$ or to $p_3$, leaving the other one unmatched. The same argument applies for the $q$-vertices on the other side of the gadget. Therefore, at least two edges from the clause gadget block~$M$. 
	
	In the remaining case, exactly one of the $b_ia_i$ edges is outside of~$M$. Since we are searching for a matching blocked by at most one edge, no further blocking edge can occur. Therefore, $p_1$, $q_1$, $p_3$ and $q_3$ are all matched in~$M$. From this point on, it is easy to see that all matchings fulfilling these requirements are $M_1, M_2$ and~$M_3$.  \myqed

Claims~\ref{cl:var} and~\ref{cl:clause} guarantee that if a matching $M$'s restriction to any of the variable or clause gadgets deviates from their special matchings, then $|bp(M)| > n_B+m_B$. 

\begin{claim}
$B$ is satisfiable if and only if $\mathcal I$ admits a matching $M$ such that $M\cap P=\emptyset$ and $|bp(M)|\leq n_B+m_B$.
\end{claim}

	\myproof Suppose we are given a satisfying truth assignment $f$ for $B$.  We construct a matching $M$ in $\mathcal I$ as follows.  In the variable gadgets, the edges of $M_T$ are chosen if the corresponding variable is $\true$ under $f$, and the edges of $M_F$ are chosen otherwise. There is at least one literal in each clause that is $\true$ under $f$.  If this literal is the $i$th in the clause, matching $M_i$ is chosen, where $i \in \{1,2,3\}$.  If more than one literal is true, we choose one of them arbitrarily.  Clearly $M$ contains no forbidden edges.  It is not difficult to verify that each gadget contributes a single blocking edge.  Thus $|bp(M)|\geq n_B+m_B$.  As a last step, we show that no interconnecting edge blocks~$M$, and thus $|bp(M)|=n_B+m_B$. Suppose that $a_i x_j$ blocks~$M$. Since it is the second choice of $a_i$, it follows that $b_i a_i \notin M$. We now know that the $i$th literal of the clause was $\true$ in the truth assignment. Therefore, $x_i$ is matched to its first choice.
	
	To prove the converse direction, we utilise Claims~\ref{cl:var} and~\ref{cl:clause}. On one hand, these two statements, together with the characteristics of the special matchings, prove that $bp(\mathcal{I}) \geq n_B + m_B$. On the other hand, $|bp(M)| = n_B + m_B$ occurs if and only if $M$'s restriction to variable gadgets is $M_T$ or $M_F$, and its restriction to clause gadgets is $M_1, M_2$ or~$M_3$. Then, assigning $\true$ to all variables with $M_T$ in their gadgets and $\false$ to the rest results in a truth assignment~$f$. Since no interconnecting edge blocks $M$, at least one literal per clause is $\true$, so $f$ satisfies~$B$. \myqed
\end{proof}

\begin{theorem}
\label{th:minbp22}
	{\sc min bp sm restricted} is solvable in $O(n)$ time if each preference list consists of at most 2 elements.
\end{theorem}

\begin{proof}
	In this constructive proof we describe an algorithm that produces an optimal matching. First, the input is simplified. Then, the graph is segmented so that each subgraph falls into a category with a specified rule for selecting the edges of an optimal matching. As in previous cases, it is sufficient to tackle {\sc min bp sm forbidden}, because instances of {\sc min bp sm forced} can be transformed to this problem.
	
	Due to the degree constraints, every component of the underlying graph is a path or a cycle. If any of these components is free of forbidden edges, then we simply fix a stable matching on it. This step is executed whenever such a component appears during the course of the algorithm. For those components with forbidden edges, we split all vertices having a first-choice forbidden edge and a second-choice edge -- unrestricted or forbidden -- into two vertices. This change does not affect $|bp(M)|$, because in this case, each edge blocks the matching if and only if its other end vertex is matched to a worse partner or is unmatched. After this splitting is executed, all components contain forbidden edges that start paths or that are inside a path, being the last choices of both of their end vertices. 
	
	Each component consists of \emph{segments} of unrestricted edges, separated by forbidden edges. When talking about a segment, we always mean a series of adjacent unrestricted edges. Since unrestricted cycles have already been eliminated by fixing a stable matching on them, every segment is a path. Due to Theorem~\ref{th:rural}, each path admits a unique stable matching. Fixing a matching on a segment induces blocking edges only among the unrestricted edges of the segment and the forbidden edges adjacent to the segment. We claim that in an optimal solution, each segment and the (at most two) forbidden edges surrounding it contribute at most two blocking edges. This is simply due to the fact that any stable solution on the unrestricted edges is blocked only by forbidden edges. Therefore, deviating from this solution might only pay off if the matching restricted to this segment is blocked by a single edge and covers both of its end vertices. 
	
	The unique stable matching $M$ on a segment $\langle v_1, v_2, ..., v_k \rangle$ falls into exactly one of the following categories:
	\begin{enumerate}
		\item $M$ covers both $v_1$ and $v_k$;
		\item $M$ covers either $v_1$ or $v_k$;
		\item $M$ covers neither $v_1$ nor $v_k$.
	\end{enumerate}
	In each step of our algorithm, a segment is chosen and a matching is fixed on it. The segment and some of the forbidden edges adjacent to it are then removed from the graph. This is done in the following way in these three cases.
	
	In case~1, an optimal solution arises from choosing~$M$. If a forbidden edge $e$ is incident to either $v_1$ or $v_k$, it cannot block~$M$.  Nor can it block any superset of $M$ in the original instance, so $e$ can be deleted. In case~2, again the optimal solution arises from choosing~$M$.  Without loss of generality suppose that $v_1$ is covered.  As in case~1, if a forbidden edge is incident to $v_1$, it cannot block a superset of $M$ in the original instance.  Now suppose that a forbidden edge $e$ is incident to~$v_k$.  Edge $e$ may block $M$, and may also block a superset of $M$ in the original instance, so it is retained. 
	
The third case is divided into two subcases, depending on whether there is a matching $M'$ that is blocked by only one edge and covers both $v_1$ and~$v_k$. Finding such a matching or proving that none exists can be done iteratively, assuming that a chosen edge is the single blocking edge and then constructing $M'$ so that no more edge blocks it. If such an $M'$ does not exist, then $M$ is chosen, and the segment (but not the forbidden edges) is removed. In the end the matching restricted to this segment will be blocked by no edge other than the two forbidden edges. Suppose $M$ is not an optimal choice in this segment. Then the optimal matching $M''$ has at most one blocking edge from this segment and its adjacent edges. $M''$ must cover at least one endpoint vertex of the segment (otherwise we have at least two blocking edges) -- without loss of generality let that vertex be~$v_1$. $M''$ cannot cover $v_k$ given the non-existence of the aforementioned matching $M'$ that was sought.  Hence the forbidden edge incident to $v_k$ contributes a blocking edge. This implies that $M''$ restricted to the segment itself is a stable matching, contradicting the fact that the unique stable matching for this segment covers neither $v_1$ nor~$v_k$.

The only remaining case is that, in the segment $\langle v_1,v_2,\dots,v_k\rangle$, the unique stable matching $M$ covers neither $v_1$ nor $v_k$, but there is a matching $M'$ blocked by one (unrestricted) edge that covers both $v_1$ and $v_k$, and this is true for all remaining segments. Let $C$ be any remaining component, which is a path comprising segments $S_1,S_2,\dots,S_r$, where $S_i$ and $S_{i+1}$ are separated by a forbidden edge ($1\leq i\leq r-1$), together with a possible additional forbidden edge at each end of the path that may be remaining.  Let $M_i$ denote a stable matching in $S_i$ (which covers neither endpoint vertex, as previously noted), and let $M_i'$ denote a matching in segment $S_i$ that covers both endpoint vertices and is blocked by at most one unrestricted edge. Now let $M''$ be an optimal matching in~$C$. We will show how to transform $M''$ to $M'''$ such that $|bp(M''')|\leq |bp(M'')|$, $M_i\cup M_{i+1}\not\subseteq M'''$ and $M'_i\cup M'_{i+1}\not\subseteq M'''$ for any $i$ ($1\leq i\leq r-1$). Firstly let $M'''=M''$.  Iteratively from $i=1$ to $r-1$ we modify $M'''$, if necessary, as follows. If $M_i\cup M_{i+1}\subseteq M''$ then replace $M_i$ by $M_i'$ in~$M'''$.  Alternatively if $M_i'\cup M_{i+1}'\subseteq M''$ then replace $M'_{i+1}$ by $M_{i+1}$ in~$M'''$. It follows that $M'''$ has the desired properties once this process terminates. Thus to reach an optimal solution in $C$ it suffices to let $M^1$ be the union of $M_1$, $M_2'$, $M_3$, $\dots$, and let $M^2$ be the union of $M_1'$, $M_2$, $M_3'$, $\dots$, and pick whichever of the two admits the fewer blocking pairs in~$C$. 
\end{proof}
	
Even with the previous two theorems, we have not quite drawn the line between tractable and hard cases in terms of vertex degrees. The complexity of {\sc min bp sm restricted} remains open for the case when preference lists are of length at most~2 on one side of the bipartite graph and are of unbounded length on the other side. However we believe that this problem is solvable in polynomial time.
	
\begin{conj}
\label{co:2_infty}
	{\sc min bp sm restricted} is solvable in polynomial time if each woman's preference list consists of at most 2 elements.
\end{conj}

\subsection{Stable Roommates problem}
\label{se:sr}

Having discussed several cases of {\sc sm}, we turn our attention to non-bipartite instances. Since {\sc sm} is a restriction of {\sc sr}, all established results on  the $\NP$-hardness and inapproximability of {\sc min bp sm restricted} carry over to the non-bipartite {\sc sr} case. As a matter of fact, more is true, since {\sc min bp sr restricted} is $\NP$-hard and difficult to approximate even if $P=\emptyset$ and $Q=\emptyset$~\cite{ABM06}.  We summarise these observations as follows.

\begin{remark}
\label{re:sr_minbp}
By Theorems~\ref{th:minbpsmforbidden} and \ref{th:inappr_minbp}, each of
{\sc min bp sr forbidden} and {\sc min bp sr forced} is $\NP$-hard and not approximable within $n^{1-\varepsilon}$, for any $\varepsilon > 0$, unless $\P=\NP$.  Moreover Theorems~\ref{th:minbp33} and \ref{th:minbpforced1} imply that each of
{\sc min bp sr forbidden} and {\sc min bp sr forced} is $\NP$-hard even if all preference lists are of length at most~3 or, in the latter case, $|Q| = 1$.  Finally {\sc min bp sr restricted} is $\NP$-hard and not approximable within $n^{\frac{1}{2}-\varepsilon}$, for any $\varepsilon > 0$, unless $\P=\NP$, even if $P=\emptyset$ and $Q=\emptyset$~\cite{ABM06}. 
\end{remark}

Remark~\ref{re:sr_minbp} already shows that Theorem~\ref{th:minbppconstant} does not carry over to the {\sc sr} case, since {\sc min bp sr forbidden} is computationally hard even if $P=\emptyset$. As for the other polynomially solvable cases, the proof of Theorem~\ref{th:minbpbpconstant} carries over without applying any modifications. Theorem~\ref{th:minbp22} also carries over to the {\sc sr} case, but it needs a slight modification. If $\deg(v) \leq 2$ for every $v \in V(G)$, then $G$ consists of paths and cycles. Each odd preference cycle without a forbidden edge contributes at least one blocking edge to any matching~\cite{Tan91} and any maximal matching on such a cycle is blocked by exactly one edge. On the remainder of the graph, the algorithm described in the proof of Theorem~\ref{th:minbp22} delivers an optimal matching for {\sc min bp sr restricted}. The following remark summarises the discussed positive results.

\begin{remark}
{\sc min bp sr restricted} is solvable in polynomial time if the minimal number of edges blocking an optimal solution is a constant or if each preference list consists of at most 2 elements.
\end{remark}


\section{Stable matchings with the minimum number of violated constraints on restricted edges}
\label{se:viol_const}

In this section, we study the second intuitive approximation concept. The desired matching is stable and violates as few constraints on restricted edges as possible. We return to our example that already appeared in Figure~\ref{fi:base}. As already mentioned earlier, the instance admits a single stable matching, namely $M = \{u_1w_1, u_2w_2, u_3w_3, u_4w_4\}$. Since $M$ contains both forbidden edges, the minimum number of violated constraints on restricted edges is~2.

This section is structured as follows: in Section~\ref{sec:gen4}, complexity and approximability results are presented for {\sc sm min restricted violations}, {\sc sr min forbidden}, {\sc sr max forced} and {\sc sr min restricted violations}.  In Section~\ref{sec:bdd4} we consider the complexity of {\sc sr min restricted violations} when the degree of the underlying graph is bounded.

\subsection{General complexity and approximability results}
\label{sec:gen4}
	As mentioned in Section~\ref{se:intro}, a weighted stable matching instance models {\sc sm min restricted violations}.

\begin{theorem}
\label{th:smminrestviol}
	{\sc sm min restricted violations} is solvable in polynomial time.
\end{theorem}

\begin{proof}
We convert {\sc sm min restricted violations} into a weighted {\sc sm} problem using the following weight function:
\[   
w(e) = 
     \begin{cases}
       -1 &\quad\text{if }e\text{ is forced},\\
       0 &\quad\text{if }e\text{ is unrestricted},\\ 
       1 &\quad\text{if }e\text{ is forbidden}.
     \end{cases}
\]
Then a stable matching $M$ has weight $|M \cap P| - |M \cap Q| = |M \cap P| + |Q \setminus M| - |Q|$.  Since $|Q|$ is invariant, a stable matching of minimum weight will violate the minimum number of constraints on restricted edges.  The polynomial-time solvability of {\sc sm min restricted violations} then follows by the fact that we can find a minimum weight stable matching in {\sc sm} in polynomial time as discussed in the last paragraph of Section~\ref{se:preliminaries}.
\end{proof}

In the {\sc sr} context, finding a minimum weight stable matching is $\NP$-hard~\cite{Fed92}, so the above technique for {\sc sm} does not carry over to {\sc sr}. Indeed special cases of {\sc sr min restricted violations} are $\NP$-hard, as the following result shows.

\begin{theorem}
\label{th:sr_min_forbidden}
	{\sc sr min forbidden dec} and {\sc sr max forced dec} are $\NP$-complete.
\end{theorem}
\begin{proof}	
Clearly both problems belong to the class $\NP$.  We begin by proving the $\NP$-hardness of {\sc sr min forbidden dec}.  To do so, we will use a reduction from the decision version of the following problem:
\begin{pr} {\sc min vx cover} \ \\
	Input: $\mathcal{I} = G$; a graph $G$ on $n$ vertices and $m$ edges.\\
	Output: A vertex cover $C \subseteq V(G)$ such that~$|C| \leq |C'|$ for every vertex cover~$C'$.
\end{pr}
\noindent Specifically, define {\sc min vx cover dec} to be the problem of deciding, given a graph $G$ and an integer $K$, whether $G$ admits a vertex cover of size at most $K$.  {\sc min vx cover dec} is $\NP$-complete~\cite{GJ79}.


Given an instance $(G, K)$ of {\sc min vx cover dec}, the following instance $(G', O, P, K)$ of {\sc sr min forbidden dec} is constructed. The entire graph $G$ is copied, and then, a gadget is attached to each vertex $v_i \in V(G)$. It is a complete bipartite graph on four vertices: one of them is $p_i = v_i$, whilst the remaining three are denoted by  $\bar{p_i}, q_i$ and~$\bar{q_i}$. Vertex $p_i$ ranks $\bar{p_i}$ at the top, followed by all $p_j$ vertices such that $v_iv_j\in E(G)$, in arbitrary order, followed by $\bar{q_i}$ at the bottom of this list. The new vertices' orderings can be seen in Figure~\ref{fi:srminforbidden}. In order to derive an instance with complete lists, all remaining vertices can be placed in arbitrary order to the bottom of the lists. Later we will see that these edges never appear in stable matchings, neither do they block them. The set of forbidden edges comprises all $p_i\bar{p_i}$ edges corresponding to the dotted grey edges in our illustrations in Figure~\ref{fi:srminforbidden}.

\begin{center}
\begin{figure}[h]
\centering
\begin{minipage}[t][5cm][c]{0,5\textwidth}	
\begin{center}
\begin{tabular}{ r c l l l}
   $p_i$:        & $\bar{p_i}$ &adjacent $p$ vertices& $\bar{q_i}$ & rest \\
   $\bar{p_i}$:  & $q_i$       & $p_i$ & rest& \\
   $q_i$:       & $\bar{q_i}$ & $\bar{p_i}$ & rest    &\\
   $\bar{q_i}$:  & $p_i$ & $q_i$ & rest    &
 \end{tabular}
\end{center}
\end{minipage}\begin{minipage}[t][5cm][c]{0,5\textwidth}
\begin{center}	
	\begin{tikzpicture}
	\node[vertex, ultra thick, gray] (p) at (0, 0) {$p$};
	\node[vertex] (p') at (3, 0) {$\bar{p}$};
	\node[vertex] (q) at (0, -3) {$q$};
	\node[vertex] (q') at (3, -3) {$\bar{q}$};
	
	\draw [ultra thick, dotted, gray] (p) -- node[edgelabel, near start] {1} node[edgelabel, near end] {2} (p');
	\draw [] (p) -- node[edgelabel, near start] {$\geq3$} node[edgelabel, near end] {1} (q');
	\draw [] (q) -- node[edgelabel, near start] {1} node[edgelabel, near end] {2} (q');
	\draw [] (q) -- node[edgelabel, near start] {2} node[edgelabel, near end] {1} (p');
\end{tikzpicture}
\end{center}
\end{minipage}

\begin{tikzpicture}
\hspace{-3mm}
\coordinate (d) at (6.5, 0);


\node[vertex] (v1) at (0, 0) {$$};
\node[vertex] (v2) at (2, 0.5) {$$};
\node[vertex] (v3) at (2.5, 2) {$$};
\node[vertex] (v4) at (-1, 1.5) {$$};
\node[vertex] (v5) at (4, 2.5) {$$};
\node[vertex] (v6) at (1, 3) {$$};

\draw [] (v1) -- (v2);
\draw [] (v1) -- (v3);
\draw [] (v2) -- (v3);
\draw [] (v4) -- (v1);
\draw [] (v1) -- (v6);
\draw [] (v3) -- (v6);
\draw [] (v5) -- (v6);

\draw[line width=1mm,-implies,double, double distance=1mm] (4,1) -- (5,1);


\node[vertex] (v1') at ($(v1) + (d)$) {$$};
\node[vertex] (v2') at ($(v2) + (d)$) {$$};
\node[vertex] (v3') at ($(v3) + (d)$) {$$};
\node[vertex] (v4') at ($(v4) + (d)$) {$$};
\node[vertex] (v5') at ($(v5) + (d)$) {$$};
\node[vertex] (v6') at ($(v6) + (d)$) {$$};

\draw [] (v1') -- (v2');
\draw [] (v1') -- (v3');
\draw [] (v2') -- (v3');
\draw [] (v4') -- (v1');
\draw [] (v1') -- (v6');
\draw [] (v3') -- (v6');
\draw [] (v5') -- (v6');

 \ktwotwo{(v1')}{-25};
 \ktwotwo{(v2')}{-30};
 \ktwotwo{(v3')}{0};
 \ktwotwo{(v4')}{120};
 \ktwotwo{(v5')}{10};
 \ktwotwo{(v6')}{120};
\end{tikzpicture}
\caption{Adding $K_{2,2}$ to each vertex of the {\sc min vx cover dec} instance}
\label{fi:srminforbidden}
\end{figure}
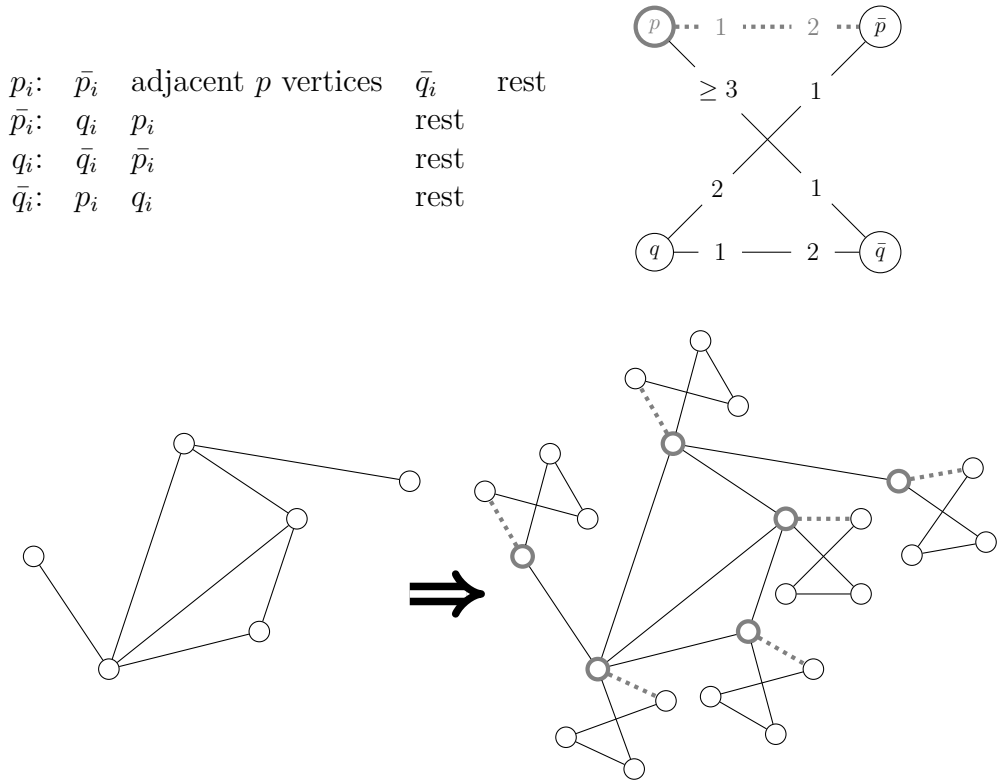
\end{center}

\begin{claim}
\label{cl:vxcover_stable}
	If $M$ is a stable matching in $G'$, then for each $i$ ($1\leq i\leq n$), either $\{p_i \bar{p_i},q_i \bar{q_i}\}\subseteq M$, or $\{p_i \bar{q_i},\bar{p_i} q_i\}\subseteq M$.
\end{claim}

	\myproof This claim follows from the structure of the introduced gadget. First, we observe that in $M$, each $\bar{p_i}$ is either matched to $p_i$ or to~$q_i$, for otherwise $p_i \bar{p_i}$ blocks~$M$, since $\bar{p_i}$ is $p_i$'s first choice. Similarly, $q_i \bar{q_i} \in M$ or $q_i\bar{p_i} \in M$, for otherwise $\bar{p_i}q_i$ would block~$M$. Finally, $p_i \bar{q_i} \in M$ or $q_i \bar{q_i} \in M$, otherwise $q_i \bar{q_i}$ blocks~$M$. These three requirements imply that the claim must hold.  \myqed

\begin{claim}
	If there is a vertex cover $C \subseteq V(G)$ with $|C| \leq K$, then there is a stable matching $M$ in $G'$ for which $|M \cap P| \leq K$.
    \label{cl:vc-stable}
\end{claim}

	\myproof The matching $M$ is constructed based on the following case distinction:	
	\begin{center}
	\begin{tabular}{ l c l }
	$\left\{p_i \bar{p_i}, q_i \bar{q_i}\right\} \subseteq M$ & if & $v_i \in C$\\
	$\left\{p_i \bar{q_i}, \bar{p_i} q_i \right\} \subseteq M$ & if & $v_i \notin C$
	\end{tabular}
\end{center}
Clearly $|M \cap P| \leq K$. Moreover, no edge in the gadgets can block $M$, because the preferences inside the gadget are cyclic. Due to the vertex cover property, edges between two $p$-vertices have at least one end vertex in $C$, thus, at least one of their end vertices is matched to its first-choice partner $\bar{p}$ in~$G'$. For each vertex in~$G'$, the edges in the sets ``rest'' are worse than the edge in~$M$.  \myqed

\begin{claim}
	If there is a stable matching $M$ in $G'$ for which $|M \cap P| \leq K$, then there is a vertex cover $C \subseteq V(G)$ with~$|C| \leq K$.
\end{claim}

	\myproof Claim~\ref{cl:vxcover_stable} allows us to investigate only two cases per gadget. We use a similar function to that in Claim~\ref{cl:vc-stable}, but in the opposite direction, in order to derive $C$ from $M$, as follows:
	
	\begin{center}
	\begin{tabular}{ l c l }
	$v_i \in C$& if & $\left\{p_i \bar{p_i}, q_i \bar{q_i}\right\} \subseteq M$ \\
	$v_i \notin C$ & if & $\left\{p_i \bar{q_i}, \bar{p_i} q_i \right\} \subseteq M$
	\end{tabular}
\end{center}
Trivially, $|C| \leq K$. Suppose $C$ is not a vertex cover. Then, there is an edge $v_i v_j = p_i p_j$ for which $\left\{ p_i \bar{q_i}, p_j \bar{q_j} \right\} \subset \left\{p_i \bar{q_i}, \bar{p_i} q_i, p_j \bar{q_j}, \bar{p_j} q_j \right\} \subseteq M$. Then $p_ip_j$ blocks~$M$.  \myqed

To prove $\NP$-completeness for {\sc sr max forced dec}, let $(G',O,Q,K)$ be the constructed instance, where $G'$, $O$ and $K$ are defined as in the {\sc sr min forbidden dec} case, and $Q = \left\{ p_i \bar{q_i}: 1\leq i \leq n\right\}$.  The correctness of the reduction then follows by a similar argument
\end{proof}

We remark that {\sc min vx cover} is $\NP$-hard and cannot be approximated within a factor of $2 - \varepsilon$ for any $\varepsilon>0$, assuming that the Unique Games Conjecture (UGC) is true~\cite{KR08}. When viewing the construction in Theorem~\ref{th:sr_min_forbidden} as a reduction from {\sc min vx cover} to {\sc sr min forbidden}, the measures of an optimal solution in each problem instance are identical.  The same is true in the case of {\sc sr max forced}.  It follows that each of {\sc sr min forbidden} and {\sc sr max forced} is not approximable within a factor of $2 - \varepsilon$ for any $\varepsilon>0$, assuming the UGC holds.

When studying {\sc sr max forced}, we measured optimality by keeping track of the number of violated constraints. One might find it more intuitive instead to maximise $|Q \cap M|$, the number of forced edges in the stable matching. Our $\NP$-hardness proof for {\sc sr max forced} given in Theorem~\ref{th:sr_min_forbidden} can be used without modification to show $\NP$-hardness under this different measure, however the  approximability results need to be revisited. In fact, this modification of the measure changes the approximability of the problem as well:

\begin{theorem}
\label{th:sr_max_forced_inappr}
For {\sc sr max forced}, the maximum of $|Q \cap M|$ cannot be approximated within $n^{\frac{1}{2}-\varepsilon}$ for any $\varepsilon > 0$, unless $\P = \NP$.
\end{theorem}

\begin{proof}
We adapt the proof of Theorem~\ref{th:sr_min_forbidden} so that the reduction is from {\sc max ind set}, the problem of finding a maximum independent set in a given graph~$G=(V,E)$.  {\sc max ind set} is not approximable within $N^{1-\varepsilon}$ for any $\varepsilon > 0$, unless $\P = \NP$~\cite{Zuc07}, where $N=|V|$. In the modified reduction 
an independent set of vertices $S$ in $G$ corresponds to the matching 
\[M=\{p_i\bar{q_i},\bar{p_i}q_i : v_i\in S\}\cup
    \{p_i\bar{p_i},q_i\bar{q_i} : v_i\notin S\}\]
in the constructed instance $\mathcal{I}$ of {\sc sr max forced}.
Suppose that $A$ is an $n^{\frac{1}{2}-\varepsilon}$-approximation algorithm that approximates $|Q \cap M_{opt}|$ in $\mathcal{I}$, for some $\varepsilon>0$, where $M_{opt}$ is an optimal solution and $n$ is the number of agents in $\mathcal{I}$.  Note that $|S_{opt}|=|M_{opt}|$, where $S_{opt}$ is a maximum independent set in~$G$.  Moreover  $n^{\frac{1}{2}-\varepsilon}=(4N)^{\frac{1}{2}-\varepsilon}\leq N^{1-\varepsilon}$ since $n=4N$ and without loss of generality we can assume that~$N\geq 4$.  We thus reach a contradiction to the inapproximability of {\sc max ind set}.
\end{proof}

\subsection{Bounded parameters}
\label{sec:bdd4}
We now turn to the complexity of {\sc sr min restricted violations} and its variants when the degree of the underlying graph is bounded or some parameter of the instance can be considered as a constant. With Theorems~\ref{th:sr_min_forbidden3} and \ref{srminrest2} we draw the line between $\NP$-hard and polynomially solvable cases in terms of degree constraints.

\begin{theorem}
\label{th:sr_min_forbidden3}
{\sc sr min forbidden dec} and {\sc sr max forced dec} are $\NP$-complete even if every preference list is of length at most~3.
\end{theorem}
\begin{proof}
As in Theorem~\ref{th:sr_min_forbidden}, both problems belong to $\NP$.  We firstly show the $\NP$-hardness of {\sc sr min forbidden dec} for this length restriction on preference lists.  To do so, we reduce from {\sc min vx cover dec} in cubic graphs, which is $\NP$-complete~\cite{GJS76,MS77}.
Hence let $\mathcal{I} = (G,K)$ be an instance of {\sc min vx cover}, where $G=(V,E)$ is a cubic graph, $V=\{v_1,\dots,v_n\}$ and $E=\{e_1,\dots,e_m\}$. For each $i$ ($1\leq i\leq n$), suppose that $v_i$ is incident to edges $e_{j_1}$, $e_{j_2}$ and $e_{j_3}$ in $G$, where without loss of generality~$j_1<j_2<j_3$.  Define $e_{i,r}=e_{j_r}$, where $r \in \{1,2,3\}$. Similarly for each $j$ ($1\leq j\leq m$), suppose that $e_j=v_{i_1} v_{i_2}$, where without loss of generality~$i_1<i_2$. Define $v_{j,s}=v_{i_s}$, where $s \in \{1,2\}$. 

We construct an instance $\mathcal{I'}$ of {\sc sr min forbidden dec} as follows. The set $V' \cup W \cup E' \cup F$ constitutes the set of vertices in~$\mathcal{I'}$, where these sets are defined as follows:
\[
\begin{array}{llllllll}
	V' & = \{v_i^r : 1\leq i\leq n\wedge r \in \{1,2,3\}\}\\
	W  & = \{w_i^r : 1\leq i\leq n\wedge r \in \{1,2,3\}\}\\
	E' & = \{e_j^s : 1\leq j\leq m\wedge s \in \{1,2\}\}\\
	F & = \{f_j^s : 1\leq j\leq m\wedge s \in \{1,2\}\}	
\end{array}
\]
The preference lists of these vertices (also indicating the edges of the graph) are as shown in Figure~\ref{preflists}. In the preference list of a vertex $v_i^r$, the symbol $e(v_i^r)$ denotes vertex $e_j^s \in E'$ such that $e_j=e_{i,r}$ and~$v_i=v_{j,s}$.  Since $i$ and $r$ are given, $e_j$ can be computed. Now we know $i$ and $j$ in the second equation, therefore we can compute~$s$. Similarly in the preference list of vertex $e_j^s$, the symbol $v(e_j^s)$ denotes vertex $v_i^r \in V'$ such that $e_j=e_{i,r}$ and~$v_i=v_{j,s}$.

Let $P=\{v_i^1 w_i^1 : 1\leq i\leq n\}$ be the set of forbidden edges in~$\mathcal{I'}$. The edges connecting $V'$ and $E'$ capture the incidence relations of the original graph $G$, while vertices in $W$ and $F$ can be seen as garbage collectors.

\begin{figure}[h]
\centering
\begin{minipage}{0.4\textwidth}
\[
\begin{array}{rll}
v_i^1: & w_i^1 ~~ e(v_i^1) ~~ w_i^2 ~~~~ & (1\leq i\leq n) \vspace{1mm}\\
v_i^2: & w_i^2 ~~ e(v_i^2) ~~ w_i^3 & (1\leq i\leq n) \vspace{1mm}\\
v_i^3: & w_i^3 ~~ e(v_i^3) ~~ w_i^1 & (1\leq i\leq n) \vspace{1mm}\\
e_j^1: & e_j^2 ~~ v(e_j^1) ~~ f_j^2 & (1\leq j\leq m) \vspace{1mm}\\
e_j^2: & f_j^1 ~~ v(e_j^2) ~~ e_j^1 & (1\leq j\leq m) \vspace{1mm}
\end{array}
\]
\end{minipage}\hspace{15mm}\begin{minipage}{0.4\textwidth}
\[
		\begin{array}{rll}
w_i^1: & v_i^3 ~~ v_i^1             & (1\leq i\leq n) \vspace{1mm}\\
w_i^2: & v_i^1 ~~ v_i^2             & (1\leq i\leq n) \vspace{1mm}\\
w_i^3: & v_i^2 ~~ v_i^3             & (1\leq i\leq n) \vspace{1mm}\\
f_j^1: & f_j^2 ~~ e_j^2             & (1\leq j\leq m) \vspace{1mm}\\
f_j^2: & e_j^1 ~~ f_j^1             & (1\leq j\leq m) \vspace{1mm}
\end{array}
\]
\end{minipage}
\caption{Preference lists in the constructed instance of {\sc sr min forbidden}.}
\label{preflists}
\end{figure}

Finally we define some further notation in~$\mathcal{I}$. For each $i$, $1\leq i\leq n$, let $V_i^c=\{v_i^r w_i^r : r \in \{1,2,3\}\}$ and let
$V_i^u=\{v_i^r w_i^{r+1} : r \in \{1,2,3\}\}$, where addition is taken modulo~3. Note that each $V_i^c$ contains exactly one forbidden edge, while $V_i^u$ has no forbidden edge. Similarly for each $j$, $1\leq j\leq m$, let $E_j^1=\{e_j^1 e_j^2,f_j^1 f_j^2\}$ and let
$E_j^2=\{e_j^1 f_j^2,e_j^2 f_j^1\}$.

\begin{claim}
	\label{cl:bp1}
	$\mathcal{I}'$ admits a stable matching in which every vertex is matched. 
\end{claim}

	\myproof Let $M=\bigcup_{i=1}^n V_i^c \cup \bigcup_{j=1}^m E_j^1$. Starting with the argument that each $v_i^r\in V'$, each $e^1_j\in E'$ and each $f^1_j\in F$ receive its first-choice partner in $M$, it is straightforward to verify that $M$ is stable: the remaining vertices prefer, to their partners, only vertices that already have their first-choice partners.  Theorem~\ref{th:rural} implies then that every stable matching in $\mathcal{I}'$ matches every vertex in~$\mathcal{I}'$. \myqed 

In Claims~\ref{cl:bp2} and~\ref{cl:bp3} we show that $G$ has a vertex cover $C$ where $|C|\leq K$ if and only if $\mathcal{I}'$ has a stable matching $M$ where~$|M\cap P|\leq K$.

\begin{claim}
\label{cl:bp2}
	If $G$ has a vertex cover $C$ such that $|C|\leq K$ in~$\mathcal{I}$, then there is a stable matching $M$ in $\mathcal{I}'$ with~$|M\cap P|\leq K$.
\end{claim}

	\myproof We construct a matching $M$ in $\mathcal{I}$ as follows. For each $i$ ($1\leq i\leq n$), if $v_i\in C$, add $V_i^c$ to $M$, otherwise add $V_i^u$ to~$M$. For each $j$ ($1\leq j\leq m$), if $v_{j,1}\in C$, add $E_j^2$ to $M$, otherwise add $E_j^1$ to~$M$. Then $|M\cap P|=|C|\leq K$.

	Now we verify that $M$ is stable in~$\mathcal{I}$. Suppose firstly that some $v_i^r\in V'$ has its third-choice partner in $M$, and prefers $e(v_i^r)$.  Then $v_i\notin C$.  Let $e_j^s=e(v_i^r)$.  Then by definition, $e_j=e_{i,r}$ and $v_i=v_{j,s}$.  Since $v_{j,s}\notin C$, $e_j$ is covered by $C$ at its other endpoint (i.e., $v_{j,3-s}$).  By construction of $M$, $E_j^s\subseteq M$.  Hence $e_j^s$ has its first-choice partner in $M$.
    
    Now suppose that some $e_j^s$ has its third-choice partner in $M$, and prefers $v(e_j^s)$.  Then $E_j^{3-s}\subseteq M$.  It follows that $v_{j,s}\in C$, since $C$ is a vertex cover.  Let $v_i^r=v(e_j^s)$.  Then by definition, $e_j=e_{i,r}$ and $v_i=v_{j,s}$.  Thus $v_i\in C$.  By construction of $M$, $V_i^c\subseteq M$.  Hence $v_i^r$ has its first-choice partner in $M$.
%
%
%
\myqed
%

\begin{claim}
	\label{cl:bp3}
	If there is a stable matching $M$ with $|M\cap P|\leq K$ in $\mathcal{I'}$, then $G$ has a vertex cover $C$ in $\mathcal{I}$ such that~$|C|\leq K$.
\end{claim}

	\myproof We construct a set of vertices $C$ in $G$ as follows. Claim~\ref{cl:bp1} states that $M$ matches every vertex in $\mathcal{I}$.  Hence for each $i$ ($1\leq i\leq n$), either $V_i^c\subseteq M$ or~$V_i^u\subseteq M$. In the former case add $v_i$ to~$C$. As $|M\cap P|\leq K$, it follows that~$|C|\leq K$. Also, for each $j$ ($1\leq j\leq m$), as $M$ matches every vertex in $\mathcal{I}$, either $E_j^1\subseteq M$ or~$E_j^2\subseteq M$. 

	Assume that $C$ is not a vertex cover in~$G$, i.e., there is an edge $e_j = v_{j,1} v_{j,2}$ such that $v_{j,1}\notin C$ and $v_{j,2}\notin C$.  Suppose that $v_{i_1}=v_{j,1}$ and $v_{i_2}=v_{j,2}$.  Then $V_{i_1}^u\subseteq M$ and $V_{i_2}^u\subseteq M$.
    
    Now let $r\in \{1,2,3\}$ be such that $e_j=e_{i_1,r}$ and let $r'\in \{1,2,3\}$ be such that $e_j=e_{i_2,r'}$.  Then $e(v_{i_1}^r)=e_j^1$ and $v_{i_1}^r$ prefers $e_j^1$ to its partner in $M$.  Similarly $e(v_{i_2}^{r'})=e_j^2$ and $v_{i_2}^{r'}$ prefers $e_j^2$ to its partner in $M$.
    
    If $E_j^1\subseteq M$ then $v_{i_2}^{r'}e_j^2$ blocks $M$.  Otherwise $E_j^2\subseteq M$ and $v_{i_1}^{r}e_j^1$ blocks $M$.  This contradiction to the stability of $M$ implies that $C$ is a vertex cover in~$G$.  \myqed
%
	
	For {\sc sr max forced}, an analogous proof can be derived if we define the set of forced edges to be $Q=\{v_i^1 w_i^2 : 1\leq i\leq n\}$.
\end{proof}

\begin{theorem}
\label{srminrest2}
	{\sc sr min restricted violations} is solvable in $O(n)$ time if every preference list is of length at most~2.
\end{theorem}

\begin{proof}
    Since the the set of matched vertices is the same in all stable matchings by Theorem~\ref{th:rural}, finding a stable matching in $O(n)$ time in these very strongly restricted instances marks all vertices that need to be matched. In each component $C$, since $C$ is a path or a cycle, there are at most two possible stable matchings satisfying these constraints. We choose the stable matching in $C$ that violates fewer constraints.\end{proof}

Short preference lists are not the only case when {\sc sr min restricted violations} becomes tractable, as our last theorem shows.
    
\begin{theorem}
\label{srminrestconst}
	{\sc sr min restricted violations} is solvable in polynomial time if the number of restricted edges or the minimal number of violated constraints is constant.
\end{theorem}

\begin{proof}
Suppose firstly that $L=|P|+|Q|$ is a constant.  We will show how to solve {\sc sr min restricted violations dec} in polynomial time.  We assume that, for the purposes of this proof, the problem definition is modified so that, given an instance $\mathcal{I} = (G, O, P, Q, K)$, we are required to find a stable matching $M$ in $G$ such that $|M\cap P|+|Q\backslash M|\leq K$, or report that no such matching exists.

Our first observation is that this problem is trivially solvable if the target value $K$ satisfies $K \geq L$.  In this case, any stable matching will suffice.

Now assume that~$K<L$.  Suppose firstly that there is a stable matching $M$ in $G$ violating $k\leq K$ restrictions.  Then $|M\cap P|=k_1$ and $|Q\backslash M|=k_2$, where $k_1+k_2=k$.  If we let $P'=M\cap P$ and $Q'=Q\backslash M$ then $M$ is a stable matching in $\mathcal I$ containing no edge in $(P\backslash P')\cup Q'$ and containing all edges in $P'\cup (Q\backslash Q')$.

Hence to solve {\sc sr min restricted violations dec} we generate all subsets $S$ of $P\cup Q$ of size $k$, for each $k\leq K$.  Then we run the algorithm of~\cite{FIM07} to determine in $O(m)$ time whether there is a stable matching containing no edge in $(P\backslash S)\cup (Q\cap S)$ and containing all edges in $(P\cap S)\cup (Q\backslash S)$.  

Thus $\sum_{i = 0}^{K}{m \choose i}=\sum_{i = 0}^{L}{m \choose i}$ subsets are generated to determine whether the desired matching exists.  The number of rounds is thus $O(m^L)$, while each round takes $O(m)$ time to complete.  The overall running time is $O(m^{L+1})$.

We now show how to use the above approach in order to solve {\sc sr min restricted violations}.  If we find a solution during course of this process then $G$ admits a stable matching $M$ such that $|M\cap P|+|Q\backslash M|\leq K$.  In order to minimise $|M\cap P|+|Q\backslash M|$ it suffices to use the above technique in combination with a binary search procedure on values of $K\leq L$.  This requires $O(\log L)$ invocations of the algorithm for the decision problem, which is a constant, and hence the overall time complexity remains $O(m^{L+1})$.
\medskip

Now suppose that $L$ is the minimal number of violated constraints, and that $L$ is constant. For each value of $K$, where $K$ starts from 0 and increases by 1 after each iteration, we execute the algorithm described above to solve {\sc sr min restricted violations dec}, noting that it is sufficient to generate all subsets of size exactly $K$ at each iteration.  We terminate as soon as we find a stable matching $M$ such that 
 $|M\cap P|+|Q\backslash M|\leq K$.  This process is bound to halt, since by definition, $\mathcal I$ admits a stable matching $M$ such that $|M\cap P|+|Q\backslash M|=L$, so $K\leq L$. Thus the overall time complexity of this approach is $O((L+1)m^{L+1})=O(m^{L+1})$, which is polynomial if $L$ is a constant.
\end{proof}

\section{Conclusion and open questions}
\label{sec:conc}
    In this paper, we investigated the stable marriage and the stable roommates problems on graphs with forced and forbidden edges. Since a solution satisfying all constraints need not exist, two relaxed problems were defined. In {\sc min bp sm restricted}, constraints on restricted edges are strict, while a matching with the  minimum number of blocking edges is searched for. On the other hand, in {\sc sr min restricted violations}, we seek stable solutions that violate as few constraints on restricted edges as possible. For both problems, we determined the complexity and studied several special cases.
    
One of the most striking open questions is the approximability of {\sc sr min restricted violations}. Even though the problem can easily be formulated as a weighted {\sc sr} problem, the 2-approximation for this latter problem~\cite{TS97,TS98} only holds for instances with specific $P$ and $Q$ sets. This is due to the non-negativity and monotonicity constraints on the 2-approximation result. 
Another open question is formulated as Conjecture~\ref{co:2_infty}: the complexity of {\sc min bp sm restricted} is not known if each woman's preference list consists of at most 2 elements.

A more general direction of further research involves the {\sc sm min restricted violations} problem. We have shown that it can be solved in polynomial time, using algorithms for minimum weight stable marriage. The following question arises naturally: is there a faster method for {\sc sm min restricted violations} that avoids reliance on Feder's algorithm or linear programming methods?

Another natural generalisation is to consider preference lists involving ties. Our hardness results carry over to this case, but the positive results need to be revisited.

Besides the two main problems discussed in this paper, other approximation concepts can also be investigated in the framework of restricted edges.  For example one alternative would be to combine the two objectives that we considered.  Can we efficiently find matchings that minimise the total number of violated constraints, that is, $|bp(M)| + |M \setminus Q| + |M \cap P|$?

Counting the number of blocking pairs is the most prevalent, but not the only relaxation of stability that has been studied in the literature. Other relaxations can also be combined with the presence of restricted edges in the instance, such as the following concepts.
\begin{itemize}
 \item \emph{Maximum internally stable matchings}~\cite{Tan90}, where the goal is to maximise the set of pairs that are stable within themselves. 
 \item \emph{Maximum irreversible stable matchings}~\cite{BIM16}, where the goal is to maximise the number of irreversible pairs. A pair is \emph{irreversible} if, once it is contained in a matching, no agent from the pair will ever be in a blocking edge, irrespective of how the outside agents are matched.
 \item \emph{Socially stable matchings}~\cite{AIKMP13}, where only a fixed subset of pairs (described by a given social network graph) can be blocking.
\end{itemize}

\section*{Acknowledgements}
We would like to thank the anonymous reviewers of this paper and an earlier version of it for their valuable comments, which helped to improve the presentation, and for suggesting several of the open problems that feature in Section \ref{sec:conc}.

\bibliographystyle{elsarticle-harv} 
\bibliography{mybib}
\end{document}